\pdfoutput=1
\RequirePackage[l2tabu, orthodox]{nag}

\documentclass[final,fleqn]{lmcs}
\pdfoutput=1

\usepackage{silence}
\usepackage{lastpage}
\lmcsdoi{15}{3}{17}
\lmcsheading{}{\pageref{LastPage}}{}{}%
{Mar.~07,~2019}{Aug.~20,~2019}{}

\usepackage[english]{babel}
\usepackage[utf8]{inputenc}

\usepackage{amssymb}
\usepackage[heavycircles]{stmaryrd}
\usepackage{microtype}

\usepackage{relsize}
\usepackage{scalerel}

\usepackage{braket}

\usepackage{marginnote}

\usepackage{booktabs}

\usepackage[fixamsmath,disallowspaces]{mathtools}

\usepackage{array}
\newcolumntype{L}{>{$}l<{$}}
 \newcolumntype{C}{>{$}c<{$}}

\usepackage[babel,autostyle=true,strict]{csquotes}
\MakeOuterQuote{"}

\usepackage[xspace]{ellipsis}

\usepackage{caption}

\newcommand*{\ie}{i.e.\@\xspace}
\newcommand*{\eg}{e.g.\@\xspace}

\newcommand*{\wloss}{w.l.o.g.\@\xspace}
\newcommand*{\Wloss}{W.l.o.g.\@\xspace}
\renewcommand{\phi}{\varphi}
\renewcommand{\epsilon}{\varepsilon}

\newcommand*{\vv}[1]{\vec{\mkern1mu#1}}

\renewcommand{\mod}[1]{\ (\mathrm{mod}\ #1)}
\newcommand*{\Nset}{\mathbb{N}}
\providecommand{\dfn}{\mathrel{\mathop:}=}
\providecommand{\ddfn}{\mathrel{\mathop{{\mathop:}{\mathop:}}}=}

\newcommand*{\poly}{\mathrm{poly}}
\newcommand*{\calL}{\mathcal{L}}
\newcommand*{\calO}{\mathcal{O}}

\newcommand*{\md}{\mathrm{md}}
\newcommand*\+{\mkern2mu}

\newcommand*{\restr}[2]{#1{\upharpoonright}#2}
\newcommand*{\opset}{\Sigma}
\newcommand*{\splits}{\mathrm{Sp}}
\newcommand*{\ssplits}{\mathrm{SSp}}
\newcommand*{\teams}{\mathrm{Tms}}

\newcommand*{\Dim}{\mathrm{Dim}}

\newcommand*{\GG}{\mathbb{G}}
\newcommand*{\occ}{\mathrm{occ}_\bor}

\newcommand*{\FS}{\mathrm{FS}}
\renewcommand{\AA}{\mathbb{A}}
\newcommand*{\BB}{\mathbb{B}}

\newcommand*{\size}{\mathrm{wd}}

\newcommand*{\PL}{\mathrm{PL}}
\newcommand*{\MTL}{\mathrm{MTL}}
\newcommand*{\ML}{\mathrm{ML}}
\newcommand*{\Prop}{\mathrm{Prop}}
\newcommand*{\lax}[1]{#1^*}
\newcommand*{\arity}[1]{\mathrm{ar}(#1)} 

\newcommand*{\bor}{\ovee}
\newcommand*{\band}{\land}
\newcommand*{\equi}{\leftrightarrow}
\newcommand*{\nequi}{\nleftrightarrow}

\newcommand*{\sor}{\mathbin{\dot{\vee}}}
\newcommand*{\sand}{\mathbin{\dot{\owedge}}}
\newcommand*{\laor}{\lor}
\newcommand*{\laand}{\owedge}
\newcommand*{\sneg}{{\sim}}
\newcommand*{\hook}{\hookrightarrow}

\newcommand*{\nempty}{\text{\textsc{ne}}}
\newcommand*{\E}{\mathsf{E}}
\DeclareMathOperator*{\bigbor}{\scalerel*{\ovee}{\sum}}

\newcommand*{\dep}[1]{{=\!\!(#1)}}
\newcommand*{\anon}{\Upsilon}

\newcommand*{\TOWERPOLY}{\mathrm{TOWER}(\mathrm{poly})} 
\newcommand*{\PSPACE}{\mathrm{PSPACE}}
\newcommand*{\ATIMEALT}[1]{\mathrm{ATIME}\text{-}\mathrm{ALT}(#1,\poly)} 

\ACMCCS{Theory of computation $\to$ Complexity theory and logic;}

\begin{document}

\title{On the succinctness of atoms of dependency}

\author[Martin Lück]{Martin Lück\rsuper{a}}
\address{\lsuper{a}Leibniz University Hannover}
\email{\texttt{lueck@thi.uni-hannover.de}}

\author[Miikka Vilander]{Miikka Vilander\rsuper{b}}
\address{\lsuper{b}Tampere University}
\email{\texttt{miikka.vilander@tuni.fi}}

\begin{abstract}
  Propositional team logic is the propositional analog to first-order team logic.
  Non-classical atoms of dependence, independence, inclusion, exclusion and anonymity can be expressed in it, but for all atoms except dependence only exponential translations are known.
  In this paper, we systematically compare their succinctness in the existential fragment, where  the splitting disjunction only occurs positively, and in full propositional team logic with unrestricted negation.
  By introducing a variant of the Ehrenfeucht--Fraïssé game called formula size game into team logic, we obtain exponential lower bounds in the existential fragment for all atoms. 
  In the full fragment, we present polynomial upper bounds also for all atoms.
\end{abstract}

\maketitle

\section{Introduction}

As a novel extension of classical logic, \emph{team semantics} provides a framework for reasoning about whole collections of entities at once, as well as their relation with each other.
Such a collection of entities is called a \emph{team}.
Originally, team semantics was introduced by Hodges~\cite{Hodges1997} to provide a compositional approach to logic of incomplete information, such as Hintikka's and Sandu's \emph{independence-friendly logic (IF-logic)}~\cite{HS89}.

In his seminal work, Väänänen~\cite{vaa07} introduced \emph{dependence logic} which extends first-order logic by so-called \emph{dependence atoms}, atomic formulas $\dep{x_1, \ldots, x_n;y}$ that intuitively express that the value of $y$ depends only on the values of $x_1,\ldots,x_n$.
While in IF-logic dependencies between variables are expressed with annotated quantifiers such as $\exists y/\{x_1,\ldots,x_n\}$, in team semantics these can be expressed without changing the quantifiers.
Accordingly, dependence logic formulas are evaluated on sets of first-order assignments (called \emph{teams}).
Besides the dependence atom, a multitude of other notions of interdependencies between variables were studied, such as the \emph{independence} of variables~\cite{GV13}, written $x_1\cdots x_n \perp y_1\cdots y_m$, the \emph{inclusion} $x_1\cdots x_n\subseteq y_1\cdots y_n$~\cite{ga12}, \emph{exclusion} $x_1\cdots x_n\mid y_1\cdots y_n$, and \emph{anonymity}
$x_1\ldots x_n \anon y_1 \ldots y_n$~\cite{Vaa}, also known as \emph{non-dependence}~\cite{RRthesis}.
We generally refer to these expressions as \emph{atoms of dependency}.
In its original formulation, dependence logic does not have a Boolean negation but only a so called dual negation $\neg$. For this negation, basic laws such as the law of the excluded middle---that either $\alpha$ or $\neg \alpha$ holds in any given interpretation---fail.
By adding a Boolean negation operator, often written $\sneg$, Väänänen~\cite{vaa07} introduced \emph{team logic} as a strictly more powerful extension of dependence logic.

In the last decade, research on logics with team semantics outside of the first-order setting has thrived as well.
A plethora of related systems has been introduced, most prominently for modal logic~\cite{vaa08}, propositional logic~\cite{Yang14,YangV16}, and temporal logic~\cite{KrebsMV15,KrebsMV018}.
Analogously to first-order team logics, variants with a Boolean negation were studied extensively~\cite{YangV17,mueller14,Kon15}.
The atoms of dependency in these logics feature a fundamental difference to their first-order counterparts:
First-order dependencies range over individuals of the universe, whereas propositional dependency atoms only range over truth values, of which there are only finitely many.
Based on this fact, unlike in first-order logic, they can be finitely defined in terms of other logical connectives.

Gogic et al.~\cite{GogicKPS95} argue that in addition to the computational complexity of a logic and which properties it can express, it is also important to consider how \emph{succinctly} the logic can express those properties.
The succinctness of especially modal and temporal logics has been an active area of research for the last couple of decades; see e.g.\ \cite{Wilke99,LutzSW01,EtessamiVW02,AdlerI03, Markey03}
for earlier work on the topic and~\cite{FrenchHIK11,FrenchHIK13,DitmarschFHI14,HoekI14}
for recent work.
A typical result states that a logic $\calL_1$ is exponentially more succinct than another logic $\calL_2$.
This means that there is a sequence of properties ${(P_n)}_{n \in \Nset}$ such that $P_n$ is definable by $\calL_1$-formulas ${(\phi_n)}_{n \in \Nset}$, but every family ${(\psi_n)}_{n \in \Nset}$ of $\calL_2$-formulas that defines ${(P_n)}_{n \in \Nset}$ is exponentially larger than ${(\phi_n)}_{n \in \Nset}$.

In team semantics, the question of succinctness has received only little attention so far.
In their paper, Hella et al.~\cite{HellaLSV14} are primarily concerned with the expressive power of modal dependence logic, but they also show that defining the dependence atom in modal logic with Boolean disjunction requires a formula of exponential size.
Similarly, Kontinen et al.~\cite{KontinenMSV17} investigate many aspects of modal independence logic and among them show that modal independence logic is exponentially more succinct than basic modal logic.
Our paper is, to our knowledge, the first systematic look at succinctness for team semantics.

\smallskip

The most commonly used systematic methods for proving succinctness results are formula size games and extended syntax trees.
Formula size games are a variant of Ehrenfeucht-Fraïssé games made to correspond to the size of formulas instead of the usual depth of some operator.
They were first introduced by Adler and Immerman~\cite{AdlerI03} for branching-time temporal logic $\mathrm{CTL}$.
The method of extended syntax trees was originally formulated by Grohe and Schweikardt~\cite{GroheS05} for first-order logic.
The notion of extended syntax tree was actually inspired by the Adler-Immerman game, and in a certain sense these two methods are equivalent: an extended syntax tree can be interpreted as a winning strategy for one of the players of the corresponding formula size game.
Both of these methods have been adapted to many languages, especially in the modal setting, see e.g.\ \cite{FrenchHIK11,HoekIK12,DitmarschFHI14}.

The formula size game we define in this paper is an adaptation of the games defined by Hella and Väänänen for propositional and first-order logic~\cite{HellaV15} and later by Hella and Vilander for basic modal logic~\cite{HellaV16}.
The new games of Hella and Väänänen are variations of the original Adler-Immerman game with a key difference.
In the original game, the syntax tree of the formula in question is constructed in its entirety and consequently the second player has an easy optimal strategy.
Thus the original game is in some sense a single player game.
The new variant uses a predefined resource that bounds the size of the constructed formula and only one branch of the syntax tree is constructed in one play.
The second player's decisions now truly matter as she gets to decide which branch that is.

\smallskip

\begin{table*}[t]\centering
 \begin{tabular}{lCCC}
  \toprule
  \text{Property} &                                       & \text{Connectives in } \Sigma & \text{Result} \\
  \midrule
  Dependence      & \mathllap{\sneg} \dep{\cdot;\cdot}    & \land,\bor, *                 & \poly         \\
                  & \dep{\cdot;\cdot}                     & \land,\bor, *                 & \exp          \\
                  & \dep{\cdot;\cdot}                     & \land,\sneg, *                & \poly         \\
  \midrule
  Independence    & \mathllap{\sneg} {\perp}\mathrlap{_c} & \land,\bor, \lor              & \poly         \\
                  & {\perp}                               & \land,\bor, *                 & \exp          \\
                  & {\perp}\mathrlap{_c}                  & \land,\sneg, *                & \poly         \\
  \midrule
  Inclusion       & \mathllap{\sneg}{\subseteq}           & \land,\bor, \lor              & \poly         \\
                  & \subseteq                             & \land,\bor, *                 & \exp          \\
                  & \subseteq                             & \land,\sneg, *                & \poly         \\
  \midrule
  Exclusion       & \mathllap{\sneg}{\mid}                & \land,\bor, *                 & \poly         \\
                  & \mid                                  & \land,\bor, *                 & \exp          \\
                  & \mid                                  & \land,\sneg, *                & \poly         \\
  \midrule
  Anonymity       & \mathllap{\sneg} \anon                & \land,\bor, \lor              & \poly         \\
                  & \anon                                 & \land,\bor, *                 & \exp          \\
                  & \anon                                 & \land,\sneg, *                & \poly         \\
  \midrule
  Parity          & \mathllap{\sneg} \oplus               & \land,\bor, *                 & \exp          \\
                  & \oplus                                & \land,\bor, *                 & \exp          \\
                  & \oplus                                & \land,\sneg, \sor             & \poly         \\
  \bottomrule
 \end{tabular}
 \captionsetup{font=small,skip=0.5em,position=bottom}
 \caption{The succinctness of team properties in propositional team logic.
  "$*$" means that the entry holds if $\lor$, $\sor$, or both are available.
  The bounds are sharp in the following sense: "poly" means that there is a polynomial translation to $\PL(\Sigma)$.
  "exp" means that there is an exponential translation to $\PL(\Sigma)$, but no sub-exponential translation.
  }\label{tab:results}
\end{table*}

\subsubsection*{Contribution.}

In this paper we consider the succinctness of atoms of dependency.
So far, it is known that these atoms can be expressed by exponentially large formulas (see Table~\ref{tab:translations}), with only the dependence atom having a known polynomial size formula~\cite{HannulaKVV18}.

In Section~\ref{sec:preliminaries} we define propositional team logic and the fragments we consider, and recall some useful known results.
In Section~\ref{sec:lower} we obtain exponential lower bounds in the \emph{existential fragment} of propositional team logic, where the splitting disjunction $\lor$ may only occur positively.
Our lower bounds imply succinctness results between logics with no atoms of dependency, and ones expanded with a single such atom. The lower bounds also show that the known translations to the existential fragment (see Table~\ref{tab:translations}) are asymptotically optimal.

Most of the lower bounds are obtained via the new formula size game for propositional team logic, including a lower bound for the parity of the cardinality of teams.
The lower bounds for dependence and exclusion atoms are obtained via the notion of \emph{upper dimension}, adapted from~\cite{HellaLSV14}.

In Section~\ref{sec:upper} we polynomially define the \emph{negations} of the considered atoms of dependency in the existential fragment.
From this, as a corollary we obtain polynomial upper bounds for full propositional team logic.
Moreover, we define parity polynomially in the full logic, even though both even and odd parities have exponential lower bounds in the existential fragment.
See Table~\ref{tab:results} for an overview of all results. For each property, the three rows correspond to defining the Boolean negation of the property with no free use of the Boolean negation operator $\sneg$, defining the property itself in the same setting, and finally defining the property with free use of Boolean negation. The required formula is classified to be either polynomial or exponential with respect to the size of the corresponding atom. We always have the Boolean disjunction $\bor$ available and either the lax disjunction $\lor$ or the strict disjunction $\sor$ or both.

Finally, we consider algorithmic applications of our results and show that the complexities of satisfiability, validity and model checking for propositional and modal team logic remain the same after extension by some atoms of dependency.

\section{Preliminaries}\label{sec:preliminaries}

\begin{defi}[Teams]
 A \emph{domain} $\Phi$ is a finite set of atomic propositions.
 A \emph{$\Phi$-assignment} is a function $s \colon \Phi \to \{0,1\}$.
 A \emph{$\Phi$-team} $T$ is a (possibly empty) set of $\Phi$-functions, $T \subseteq \Phi \to \{0,1\}$.
 The set of all $\Phi$-teams is denoted by $\teams(\Phi)$.
\end{defi}

\begin{defi}[Splits]
 Let $T$ be a team.
 We say that an ordered pair $(T_1, T_2)$ of teams is a \emph{split of $T$}, if $T_1, T_2 \subseteq T$ and $T_1 \cup T_2 = T$.
 We say that a split $(T_1, T_2)$ is \emph{strict} if $T_1 \cap T_2 = \emptyset$.
 Otherwise it is \emph{lax}.
 We denote the set of splits of $T$ by $\splits(T)$, and the set of strict splits of $T$ by $\ssplits(T)$.
\end{defi}

\begin{defi}[$\PL(\Sigma,\Phi)$-formulas]
  Let $\Sigma$ be a set of connectives $\circ$ each with a designated arity $\arity{\circ} \geq 0$.
 A $\Phi$-\emph{literal} is a string of the form $\top$, $\bot$, $\sneg \top$, $\sneg \bot$, $p$, $\neg p$, $\sneg p$, or $\sneg \neg p$, where $p \in \Phi$.
 The set of $\PL(\Sigma,\Phi)$-formulas is then the smallest set containing all $\Phi$-literals and closed under connectives in $\Sigma$, \ie, if $\phi_1,\ldots,\phi_n \in \PL(\Sigma,\Phi)$ and $\arity{\circ} = n$, then $\circ(\phi_1,\ldots,\phi_n) \in \PL(\Sigma,\Phi)$.
\end{defi}

Note that when we consider a logic with free usage of Boolean negation in front of arbitrary formulas, we include $\sneg$ in the set $\Sigma$. Otherwise, we always allow the Boolean negation $\sneg$ to occur in literals. In our setting the usual empty team property of every formula being true on the empty team, fails. We motivate this choice below after Proposition~\ref{prop:expcomp}.

Let $\Prop(\phi) \subseteq \Phi$ denote the set of propositional variables that occur in the formula $\phi$.
We will omit the domain $\Phi$ if it is clear from the context or makes no difference, and write only $\PL(\Sigma)$.
We consider the following connectives:
\begin{alignat*}{3}
  & T \vDash \top               &  & \text{ always,}                                                                                   \\
  & T \vDash \bot               &  & \Leftrightarrow\; T = \emptyset                                                                   \\
  & T \vDash p                  &  & \Leftrightarrow\; \forall s \in T \colon s(p) = 1\text{,}                                         \\
  & T \vDash \neg p             &  & \Leftrightarrow\; \forall s \in T \colon s(p) = 0\text{,}                                         \\
  & T \vDash \sneg \psi         &  & \Leftrightarrow\;T \nvDash \psi\text{,}                                                           \\
  & T \vDash \psi \land \theta  &  & \Leftrightarrow\;T \vDash \psi \text{ and }T \vDash \theta\text{,}                                \\
  & T \vDash \psi \bor \theta   &  & \Leftrightarrow\;T \vDash \psi \text{ or }T \vDash \theta\text{,}                                 \\
  & T \vDash \psi \lor \theta   &  & \Leftrightarrow\;\exists (S, U) \in \splits(T) : S \vDash \psi\text{ and }U\vDash \theta\text{,}  \\
  & T \vDash \psi \sor \theta   &  & \Leftrightarrow\;\exists (S, U) \in \ssplits(T) : S \vDash \psi\text{ and }U\vDash \theta\text{,} \\
  & T \vDash \psi \laand \theta &  & \Leftrightarrow\;\forall (S, U) \in \splits(T) : S \vDash \psi\text{ or }U\vDash \theta\text{,}   \\
  & T \vDash \psi \sand \theta  &  & \Leftrightarrow\;\forall (S, U) \in \ssplits(T) : S \vDash \psi\text{ or }U\vDash \theta\text{,}
\end{alignat*}

Note that, as usually in the context of team logic, we have two different negations: a \emph{dual negation} $\neg$ and a \emph{contradictory negation} $\sneg$.
For example, we have the equivalences $\neg(p \lor q) \equiv \neg p \land \neg q$ and $\sneg(p \lor q) \equiv \sneg p \laand \sneg q$, but $\neg p \land \neg q \not\equiv \sneg p \laand \sneg q$.
Also note that in team logic we have four different logical constants, namely $\top = \neg \bot$ (always true), $\sneg \top$ (always false), $\bot = \neg \top$ (true in the empty team) and $\sneg \bot$ (true in non-empty teams).

\smallskip

We say $\phi$ \emph{entails} $\psi$, in symbols $\phi \vDash \psi$, if $T \vDash \phi$ implies $T \vDash \psi$ for all domains $\Phi \supseteq \Prop(\phi) \cup \Prop(\psi)$ and $\Phi$-teams $T$.
If $\phi \vDash \psi$ and $\psi \vDash \phi$, then we write $\phi \equiv \psi$ and say that $\phi$ and $\psi$ are \emph{equivalent}.

\begin{defi}
  A $\PL(\{\land,\lor\})$-formula that contains no $\sneg$ is a \emph{purely propositional} formula.
\end{defi}
We will consistently use the letters $\alpha,\beta,\gamma,\ldots$ for purely propositional formulas, whereas $\phi,\psi,\theta,\ldots$ will denote arbitrary formulas.

We define the shorthands $\nempty \dfn \sneg \bot$, which defines non-emptiness of teams, and $\E \alpha \dfn \top \lor (\nempty \land \alpha)$, which expresses that at least one assignment in the team satisfies the purely propositional formula $\alpha$.

Many formulas of team logic enjoy useful closure properties:

\begin{defi}
  Let $\phi$ be a $\PL(\Sigma,\Phi)$-formula.
 \begin{itemize}
  \item $\phi$ is \emph{union closed} if, for any set of $\Phi$-teams $\mathcal{T}$ such that $\forall T \in \mathcal{T} : T \vDash \phi $ we have $\bigcup\mathcal{T} \vDash \phi$.
  \item $\phi$ is \emph{downward closed} if, for any $\Phi$-teams $T_1,T_2$, if $T_2 \vDash \phi$ and $T_1 \subseteq T_2$, we have $T_1 \vDash \phi$.
  \item $\phi$ is \emph{upward closed} if, for any $\Phi$-teams $T_1,T_2$, if $T_2 \vDash \phi$ and $T_1 \supseteq T_2$, we have $T_1 \vDash \phi$.
  \item $\phi$ has the \emph{empty team property} if $\emptyset \vDash \phi$.
  \item $\phi$ is \emph{flat} if, for any $\Phi$-team $T$, $T \vDash \phi$ if and only if $\{s\} \vDash \phi$ for all $s \in T$.
 \end{itemize}
\end{defi}

\noindent
A formula is flat if and only if it is union closed, downward closed, and has the empty team property.

\begin{prop}\label{prop:strict-lax}
 Let $\phi,\psi \in \PL(\Sigma)$ such that at least one of $\phi$ and $\psi$ is downward closed.
 Then $\phi \lor \psi \equiv \phi \sor \psi$.
\end{prop}
\begin{proof}
 Obviously, $\phi \sor \psi$ entails $\phi \lor \psi$.
 Conversely, if $T \vDash \phi \lor \psi$ via some split $(T_1,T_2)$ of $T$, then either $T_1 \setminus T_2$ will still satisfy $\phi$ or $T_2 \setminus T_1$ will satisfy $\psi$.
 So either $(T_1 \setminus T_2, T_2)$ or $(T_1, T_2 \setminus T_1)$ is a strict split of $T$ witnessing $\phi \sor \psi$.
\end{proof}

\begin{prop}
 Every $\sneg$-free $\PL(\{\land,\lor,\sor\})$-formula is flat.
 In particular, every purely propositional formula is flat.
\end{prop}
\begin{proof}
 An easy inductive proof.
\end{proof}

An important property of propositional (and other) logics is \emph{locality}, which means that formulas depend only on the assignment to variables that actually occur in the formula.
This property can be generalized to team semantics.

\begin{defi}
 If $T$ is a $\Psi$-team and $\Phi \subseteq \Psi$, the \emph{projection of $T$ onto $\Phi$}, denoted $\restr{T}{\Phi}$, is defined as the $\Phi$-team $\Set{ \restr{s}{\Phi} | s \in T}$, where $\restr{s}{\Phi}$ is the the restriction of the function $s$ to the domain $\Phi$.
\end{defi}

\begin{defi}
 A formula $\phi \in \PL(\Sigma,\Phi)$ is \emph{local} if, for any domain $\Psi \supseteq \Phi$ and $\Psi$-team $T$, it holds $T \vDash \phi$ if and only if $\restr{T}{\Phi} \vDash \phi$.
\end{defi}

\begin{propC}[\cite{YangV17}]\label{prop:locality}
 Every $\PL(\{\land,\sneg,\lor\})$-formula is local.
\end{propC}

Note that locality quickly fails if we admit strict splitting $\sor$ (cf.~\cite{YangV17}).
The formula $\psi \dfn \sneg p \sor \sneg p \sor \sneg p$ is an easy counter-example to the locality of $\PL(\{\sor\})$.\label{p:local-counter-example}
No team with domain $\{p\}$ does satisfy $\psi$, since it needs at least three assignments in the team, but for example the maximal $\{p,q\}$-team satisfies $\psi$.

\smallskip

\begin{defi}[Satisfiability]
  A formula $\phi$ is \emph{$\Phi$-satisfiable} if $T \vDash \phi$ for at least one $\Phi$-team $T$.
\end{defi}

The domain is crucial here:
The previous example formula $\psi$ is $\{p,q\}$-satisfiable, but not $\{p\}$-satisfiable.

Often the empty team is excluded in the definition of satisfiability, especially in logics with the empty team property where otherwise every formula would be satisfiable.
This is not necessary here as these definitions are interchangeable; $\phi$ is satisfiable in a non-empty team iff $\phi \land \nempty$ is satisfiable, and $\phi$ is satisfiable iff $\top \lor \phi$ is satisfiable in a non-empty team.

Usually, for propositional team logic, $\laand$ and $\sand$ are omitted since they are definable
as $\phi \laand \psi \equiv \sneg(\sneg \phi \lor \sneg \psi)$, and $\phi \sand \psi \equiv \sneg(\sneg \phi \sor \sneg \psi)$.
If they are removed entirely, then the splitting disjunction may occur only positively, that is, splits of team may only be quantified existentially.
This fragment plays an important role in the paper.

\begin{defi}
The \emph{existential fragment} is $\PL(\{\land,\bor,\lor,\sor\})$.
\end{defi}

It is well known that team logic is inherently second-order in nature:
First-order dependence logic is actually equivalent to existential second-order logic~\cite{vaa07}, and equivalent to full second-order logic if arbitrary negation is added~\cite{KN09}.
In the same vein, propositional team logic is equivalent to second-order logic over $\{0,1\}$, and to existential second-order logic if $\sneg$ is restricted~\cite{HannulaKLV16}.
In all these results, the splitting disjunction $\lor$ simulates set quantification.
From this perspective, we call the fragment with only positive $\lor$ "existential".

Note that, unlike in the first-order setting, for propositional logics the difference between existential and full logic emerges only in succinctness, not in expressive power.
Indeed Yang and Väänänen~\cite{YangV17} showed that already the existential fragment is expressively complete:

\begin{prop}\label{prop:expcomp}
  For every set $P$ of $\Phi$-teams there is a formula $\phi$ in the existential fragment such that $T \in P \Leftrightarrow T \vDash \phi$ for all $\Phi$-teams $T$.
  In particular, for every $\Sigma$ and formula $\psi \in \PL(\Sigma,\Phi)$ there is a formula $\phi$ of the existential fragment such that $\psi \equiv \phi$.
\end{prop}

Essentially this is the reason we keep Boolean negation in literals.
While expressively complete, the fragment lacks the succinctness of full propositional team logic with free use of Boolean negation.
For this reason, we find the existential fragment to be a suitable logic to compare in terms of succinctness to full propositional team logic.

We proceed with the definition of the \emph{size} of a formula.
The literature contains many different accounts of what should be considered formula size.
We take as our basic concept the length of the formula as a string.
Since in team semantics the domain is often fixed and finite, we consider each proposition symbol to be only one symbol in the string.
In Section~\ref{sec:lower} we define another measure of formula size called width because it is more convenient for the formula size game. Since we only use width for lower bounds and length is always greater than width, we refer to length in the theorems for the lower bounds.

\begin{defi}
 The \emph{length} of a formula $\phi \in \PL(\opset)$, denoted by $|\phi|$, is the length of $\phi$ as a string, counting proposition symbols as one symbol.
\end{defi}

\smallskip

If $\alpha$ is a purely propositional formula and not an atomic proposition, then technically $\neg\alpha$ is not a formula; then by $\neg \alpha$ we refer to the formula that is obtained from $\alpha$ by pushing $\neg$ inwards using classical laws, \ie,
$\neg (\beta \land \gamma) \dfn (\neg \beta \lor \neg \gamma)$ and $\neg (\beta \lor \gamma) \dfn (\neg \beta \land \neg \gamma)$.
For tuples $\vv{\alpha} = (\alpha_1,\ldots,\alpha_n)$ and $\vv\beta = (\beta_1,\ldots,\beta_n)$ of purely propositional formulas, we write $\vv{\alpha} \equi \vv{\beta}$ for the formula $\bigwedge_{i=1}^n ((\alpha_i \land \beta_i) \lor (\neg \alpha_i \land \neg \beta_i))$
and $\vv{\alpha} \nequi \vv\beta$ for $\neg(\vv\alpha \equi \vv\beta)$.
Note that since the formula $\vv\alpha \equi \vv\beta$ is purely propositional, we may use the dual negation $\neg$ here.

By slight abuse of notation, we will write $s(\alpha)$ even if $\alpha$ is not an atomic proposition, and mean
\[
 s(\alpha) = \begin{cases}1 & \text{if }\{s\} \vDash \alpha\text{,} \\
  0 & \text{else.}
 \end{cases}
\]
Finally, $s(\vv\alpha)$ is short for the vector $(s(\alpha_1),\ldots,s(\alpha_n)) \in {\{0,1\}}^n$.

\smallskip%
\label{p:atoms}

We consider the following atoms of dependency, where $\vv\alpha$, $\vv\beta$, $\vv\gamma$ are (possibly empty) tuples of formulas:
\begin{description}
 \item[Dependence] $\dep{\vv\alpha;\vv\beta}$: \begin{align*}
        \makebox[3cm]{$T \, \vDash \, \dep{\vv\alpha;\vv\beta} \; $} \Leftrightarrow \; \forall s,s' \in T : s(\vv\alpha)=s'(\vv\alpha)  \Rightarrow s(\vv\beta) = s'(\vv\beta) 
       \end{align*}
 \item[Independence] $\vv\alpha \perp \vv\beta$: \begin{align*}
        \makebox[3cm]{$T \, \vDash \, \vv\alpha \perp \vv\beta \;$} \Leftrightarrow \;\forall s, s' & \in  T : \exists s'' \in T :  s(\vv\alpha) = s''(\vv\alpha) \text{ and }s'(\vv\beta) = s''(\vv\beta) 
       \end{align*}
 \item[Conditional independence] $\vv\alpha \perp_{\vv\beta} \vv\gamma$: \begin{align*}
        \makebox[3cm]{$T \, \vDash \, \vv\alpha \perp_{\vv\beta} \vv\gamma \;$} \Leftrightarrow \;\forall s, s' & \in  T : \text{ if } s(\vv\beta) = s'(\vv\beta) \text{ then }                                                 \\ 
        & \exists s'' \in T :  s(\vv\alpha\vv\beta) = s''(\vv\alpha\vv\beta) \text{ and }s'(\vv\gamma) = s''(\vv\gamma)
       \end{align*}
 \item[Inclusion] $\vv\alpha \subseteq \vv\beta$, where $\vv\alpha$ and $\vv\beta$ have equal length:
       \begin{align*}
         & \makebox[3cm]{$T \, \vDash \, \vv\alpha \subseteq \vv\beta \; $}\Leftrightarrow \; \forall s \in T \, \exists s' \in T : s(\vv\alpha) = s'(\vv\beta) 
       \end{align*}
 \item[Exclusion] $\vv\alpha \mid \vv\beta$, where $\vv\alpha$ and $\vv\beta$ have equal length:
       \begin{align*}
         & \makebox[3cm]{$T \, \vDash \, \vv\alpha \mid \vv\beta \; $} \Leftrightarrow \; \forall s \in T \, \forall s' \in T : s(\vv\alpha) \neq s'(\vv\beta) 
       \end{align*}
 \item[Anonymity] $\vv\alpha \anon \vv\beta$:
       \begin{align*}
         & \makebox[3cm]{$T \, \vDash \, \vv\alpha \anon \vv\beta \;$} \Leftrightarrow \; \forall s \in T \, \exists s' \in T :  s(\vv\alpha) = s'(\vv\alpha) \text{ and } s(\vv\beta) \neq s'(\vv\beta) 
       \end{align*}
\end{description}

\noindent
Originally, the dependence and independence atoms were introduced in the first-order setting by Väänänen~\cite{vaa07} and Grädel and Väänänen~\cite{GV13}.
Inclusion and exclusion were considered by Galliani~\cite{ga12}.
The anonymity atom is due to Väänänen~\cite{Vaa}.
The propositional counterparts we study here, except for the anonymity atom, were first studied by Yang~\cite{Yang14}.

\begin{prop}
  Let $\Sigma = \{\land,\bor,\lor\}$ or $\Sigma = \{\land,\bor,\sor\}$.
  The atoms of dependence, conditional independence, inclusion, exclusion and anonymity are expressible by $\PL(\Sigma)$-formulas of size $2^{\calO(n)}$.
\end{prop}
\begin{proof}
  See Table~\ref{tab:translations} for $\PL(\{ \land,\bor,\lor \})$-formulas defining each atom. We prove the case of the inclusion atom and leave the rest to the reader.

  Let $\phi$ be the defining formula of Table~\ref{tab:translations} for the inclusion atom and let $\vv\alpha = (\alpha_1, \dots, \alpha_n)$ and $\vv\beta = (\beta_1, \dots, \beta_n)$ be tuples of purely propositional formulas.
  Assume $T \vDash \vv\alpha \subseteq \vv\beta$ and let $\vv{c} \in {\{\top, \bot\}}^n$.
  If there is an assignment $t \in T$ such that $t(\vv\alpha) = t(\vv{c})$, then by the inclusion atom there is another assignment $t' \in T$ such that $t(\vv{c}) = t'(\vv{c}) = t'(\vv\beta)$.
  Thus $\E(\vv\beta = \vv{c})$ holds.
  If there is no such assignment $t$, then $T \vDash \vv\alpha \neq \vv{c}$ holds.
  For every $\vv{c}$ the Boolean disjunction $(\vv\alpha \neq \vv{c}) \bor \E(\vv\beta = \vv{c})$ holds, so $T \vDash \phi$.

  Conversely, assume $T \vDash \phi$. Let $t \in T$ be an assignment. Let $t(\vv\alpha) = \vv{b} \in {\{0,1\}}^n$, and let $\vv{s} \in {\{\top,\bot\}}^n$ such that $t(\vv{s}) = \vv{b}$.
  Now clearly $\vv\alpha \neq \vv{s}$ does not hold so $\E(\vv\beta = \vv{s})$ holds.
  Consequently, there is an assignment $t' \in T$ such that $s'(\beta) = \vv{b} = s(\vv\alpha)$, so $T$ satisfies the inclusion atom.

  For $\PL(\{ \land,\bor, \sor\})$, it is easy to check that replacing each occurrence of $\lor$ with $\sor$ leads to an equivalent formula.
\end{proof}

\begin{table}[t]
\centering
\begin{tabular}{CCL}
\toprule
\dep{\vv \alpha; \vv\beta} & \equiv &
\begin{aligned}
  \bigvee_{\vv{c}\+\in \{\top,\bot\}^{n}} ((\vv\alpha \equi \vv{c}) \land \bigwedge_{i=1}^m (\beta_i \bor \neg \beta_i))
\end{aligned}\\[6mm]
\vv{\alpha} \perp_{\vv{\gamma}} \vv{\beta} & \equiv &
\begin{aligned}
  \bigvee_{\vv{c}\+ \in \{\top,\bot\}^k} ((\vv\gamma \equi \vv c) \land (\vv{\alpha} \perp \vv{\beta}))
\end{aligned}\\[6mm]
\vv{\alpha} \perp \vv{\beta} & \equiv &
\begin{aligned}
  \bigwedge_{\substack{\vv{c}\+ \in \{\top,\bot\}^{n}\\\vv{c}\+' \in \{\top,\bot\}^m }} (\vv{\alpha} \nequi \vv{c}) \bor (\vv{\beta} \nequi \vv{c}\+') \bor \E \big((\vv{\alpha} \equi \vv{c}) \land (\vv{\beta} \equi \vv{c}\+')\big)
\end{aligned}\\[6mm]
\vv{\alpha} \subseteq \vv{\beta} & \equiv &
\begin{aligned}
  \bigwedge_{\vv{c} \+\in \{\top,\bot\}^{n}} (\vv{\alpha} \nequi \vv{c}) \bor \E (\vv{\beta} \equi \vv{c})
\end{aligned}\\[6mm]
\vv{\alpha} \mid \vv{\beta} & \equiv &
\begin{aligned}
  \bigwedge_{\vv{c} \+\in \{\top,\bot\}^{n}} (\vv{\alpha} \nequi \vv{c}) \bor (\vv{\beta} \nequi \vv{c})
\end{aligned}\\[6mm]
\vv{\alpha} \anon \vv{\beta} & \equiv &
\begin{aligned}
  \bigvee_{\vv{c}\+ \in \{\top,\bot\}^{n}} \big(\vv\alpha \equi \vv{c} \land \bigbor_{i=1}^m (\E \beta_i \land \E \neg \beta_i)\big)
\end{aligned}\\[6mm]
\bottomrule
\end{tabular}
 \captionsetup{font=small,skip=0.5em,position=bottom}
\caption{Exponential translations of atoms in the existential fragment, where $\vv\alpha = (\alpha_1,\ldots,\alpha_n)$, $\vv\beta = (\beta_1,\ldots,\beta_m)$ (with $n = m$ for $\subseteq$ and $\mid$) and $\vv\gamma = (\gamma_1,\ldots,\gamma_k)$.\label{tab:dep-atoms}
}\label{tab:translations}
\end{table}

\section{Exponential lower bounds for team properties}\label{sec:lower}

Though the length of a formula is the most immediate measure of formula size, it is not the most practical one in terms of defining a formula size game. For a measure better suited to the game we have chosen the number of literals in a formula, which we call width. As the name suggests, width corresponds to the number of leaves in the syntax tree of the formula.

\begin{defi}
	The \emph{width} of a formula $\phi \in \PL(\opset)$, denoted by $\size(\phi)$, is defined recursively as follows:
	\begin{itemize}
		\item $\size(l) = 1$ for a literal $l$,
		\item $\size(\psi \circ \theta) = \size(\psi) + \size(\theta)$, where $\circ \in \opset$ is binary,
		\item $\size(\circ \psi) = \size(\psi)$, where $\circ \in \opset$ is unary.
	\end{itemize}
\end{defi}
\noindent For the actual upper and lower bounds we prove, the difference between length and width is inconsequential. The number of binary connectives, and therefore parentheses, depends on the number of literals and the number of negations of either kind for a minimal formula is bounded by the number of literals. Note that for the game we also assume formulas to be in negation normal form, but this doesn't affect the width of formulas.

\subsection{A formula size game for team semantics}

Let $\AA_0$ and $\BB_0$ be sets of $\Phi$-teams and let $k_0$ be a natural number.
Let $\opset \subseteq \{\bor, \band, \laor, \laand, \sor, \sand\}$ be a set of connectives.
Note that if the strong negation $\sneg$ is freely available in the fragment under consideration, then either none or both of each pair of dual operators must be included in $\opset$.

The formula size game $\FS^\opset_{k_0}(\AA_0, \BB_0)$ for $\PL(\opset)$ has two players, S (Samson) and D (Delilah).
Positions of the game are of the form $(k, \AA, \BB)$, where $\AA$ and $\BB$ are sets of teams and $k$ is a natural number.

The goal of S is to construct a formula $\phi$ that \emph{separates} $\AA$ from $\BB$, which means that $T \vDash \phi$ for every team $T \in \AA$, denoted $\AA \vDash \phi$ and $T \nvDash \phi$ for every team $T \in \BB$, denoted $\BB \vDash \sneg \phi$. Note that $\BB \vDash \sneg \phi$ is different from $\BB \nvDash \phi$ since the first states that no team in $\BB$ satisfies $\phi$ and the second only that not all teams in $\BB$ satisfy $\phi$.

The starting position is $(k_0, \AA_0, \BB_0)$.
If $k_0 = 0$, D wins the game.
In a position $(k, \AA, \BB)$ with $k \geq 1$, S must make one of $|\Sigma|+1$ moves to continue the game.
The available moves are the ones given by $\Sigma$ and the literal move.
The moves work as follows:
\begin{itemize}
	\item $\bor$-move: S chooses subsets $\AA_1, \AA_2 \subseteq \AA$ such that $\AA_1 \cup \AA_2 = \AA$ and natural numbers $k_1, k_2 > 0$ such that $k_1 + k_2 = k$.
Then D chooses $i \in \{1,2\}$.
The game continues from the position $(k_i, \AA_i, \BB)$.

	\item $\band$-move: Same as the $\bor$-move with the roles of $\AA$ and $\BB$ switched.

	\item $\laor$-move: For every team $A \in \AA$, S chooses a split $(A_1, A_2)$.
Let $\AA_i = \{A_i \mid A \in \AA\}$ for $i \in \{1,2\}$.
For every team $B \in \BB$, S chooses a function $f_B: \splits(B) \to \{1,2\}$.
Let $\BB_i = \{B_i \mid f_B(B_1, B_2) = i, (B_1, B_2) \in \splits(B), B \in \BB \}$ for $i \in \{1,2\}$.
Finally, S chooses natural numbers $k_1, k_2 > 0$ such that $k_1 + k_2 = k$.
Then D chooses a number $i \in \{1,2\}$.
The game continues from the position $(k_i, \AA_i, \BB_i)$.

	\item $\laand$-move: Same as the $\laor$-move with the roles of $\AA$ and $\BB$ switched.

	\item $\sor$-move: Same as the $\laor$-move except all splits $(A_1, A_2)$ and $(B_1, B_2)$ considered are strict.

	\item $\sand$-move: Same as the $\sor$-move with the roles of $\AA$ and $\BB$ switched.

	\item Literal move: S chooses a $\Phi$-literal $l$.
	If $l$ separates $\AA$ from $\BB$, S wins.
	Otherwise, D wins.
\end{itemize}
While the definition of the $\lor$-move is quite technical, the intuition is well grounded in the semantics of the connective $\lor$. Let us assume S has a formula in mind with $\lor$ as the outermost connective. On the $\AA$-side S simply splits each team $A$ into two teams, $A_1$ and $A_2$, such that $A_1$ satisfies the left disjunct and $A_2$ satisfies the right one. The $\BB$-side is more involved. S claims that no team in $\BB$ satisfies the disjunction so for each team $B$ and each split of that team, $(B_1, B_2)$, he must choose which $B_i$ does not satisfy the corresponding disjunct. These choices are gathered in the function $f_B$ for each team. Finally $\BB_1$ gathers all of the teams S has claimed to not satisfy the first disjunct, and the same for $\BB_2$ and the second disjunct.

The number $k$ can be considered a resource for S in the following sense. Since for all the connective moves $k_1, k_2 > 0$, the number $k$ decreases in each move, and if $k = 1$, only the literal move is available. Thus, in a finite number of moves, S expends his resource $k$ and must eventually make a literal move which will end the game and one of the players will win.

We first prove that winning strategies for the formula size game $\FS^\opset_{k_0}(\AA_0, \BB_0)$ correspond to $\PL(\opset)$-formulas of size at most $k_0$ that separate $\AA_0$ from $\BB_0$.

\begin{thm}\label{gametheorem}
	Let $\AA_0$ and $\BB_0$ be sets of teams and let $k_0 \in \Nset$.
Then the following conditions are equivalent:
	\begin{enumerate}[10000]
		\item[${(1)}_{k_0}$] S has a winning strategy for the game $\FS^\opset_{k_0}(\AA_0, \BB_0)$.
		\item[${(2)}_{k_0}$] There is a formula $\phi \in \PL(\opset)$ with $\size(\phi) \leq k_0$ which separates $\AA_0$ from $\BB_0$.
	\end{enumerate}
\end{thm}
\begin{proof}
	We prove the equivalence of ${(1)}_{k_0}$ and ${(2)}_{k_0}$ by induction on $k_0$.

	Let $k_0 = 1$. The only type of move available for S is the literal move, so S if has a winning strategy, then there is a literal that separates $\AA_0$ from $\BB_0$. Conversely, the only formulas with size at most 1 are literals so if such a formula exists, then S wins by choosing that formula for a literal move.

	Let $k_0 > 1$ and assume that the equivalence of ${(1)}_k$ and ${(2)}_k$ holds for all natural numbers $k < k_0$ and all sets of teams $\AA$ and $\BB$.

    \smallskip
	${(1)}_{k_0} \Rightarrow {(2)}_{k_0}$:
	Let $\delta$ be a winning strategy of $S$ for the game $\FS^\opset_{k_0}(\AA_0, \BB_0)$.
	We divide the proof into cases according to the first move of $\delta$.
	We handle all operators possibly in $\opset$ except for dual cases.
	\begin{itemize}
		\item Literal move: Since S is playing according to the winning strategy $\delta$, the literal $l$ chosen by S separates $\AA_0$ from $\BB_0$.
		In addition, $\size(l) = 1 \leq k_0$.

		\item $\bor$-move: Let $(k_1, \AA_1, \BB_0)$ and $(k_2, \AA_2, \BB_0)$ be the successor positions chosen by S according to $\delta$.
		Since $\delta$ is a winning strategy, S has a winning strategy for both games $\FS^\opset_{k_i}(\AA_i, \BB_0)$.
		By induction hypothesis, there are formulas $\psi_i$ with $\size(\psi_i) \leq k_i$ that separate $\AA_i$ from $\BB_0$.
		Let $\phi = \psi_1 \bor \psi_2$.
		We have $\AA_0 = \AA_1 \cup \AA_2$ so $\AA_0 \vDash \phi$.
		On the other side we have $\BB_0 \vDash \sneg \psi_1$ and $\BB_0 \vDash \sneg \psi_2$ so $\BB \vDash \sneg \phi$.
		Finally $\size(\phi) = \size(\psi_1) + \size(\psi_2) \leq k_1 + k_2 = k_0$.

		\item $\laor$-move: Let $(k_1, \AA_1, \BB_1)$ and $(k_2, \AA_2, \BB_2)$ be the successor positions chosen by S according to $\delta$.
		Again by induction hypothesis there are formulas $\psi_i$ with $\size(\psi_i) \leq k_i$ which separate $\AA_i$ from $\BB_i$.
		Let $\phi = \psi_1 \laor \psi_2$.
		For each $A \in \AA_0$ S chose a split $(A_1, A_2)$.
		Now $A_1 \vDash \psi_1$ and $A_2 \vDash \psi_2$ so $A \vDash \phi$.
		On the other side, for each $B \in \BB_0$, S chose a function $f_B : \splits(B) \to \{1,2\}$.
		For each $(B_1, B_2) \in \splits(B)$, if $f_B(B_1, B_2) = i$, then $B_i \nvDash \psi_i$.
		Thus $B \nvDash \phi$.
		The width of $\phi$ is as in the previous case.

		\item $\sor$-move: Same as the $\laor$-move except all splits considered are strict.
	\end{itemize}

    \smallskip
	${(2)}_{k_0} \Rightarrow {(1)}_{k_0}$:
	Let $\phi \in \PL(\opset)$ with $\size(\phi) \leq k_0$ which separates $\AA_0$ from $\BB_0$.
	We give the first move of the winning strategy of S for the game $\FS^\opset_{k_0}(\AA_0, \BB_0)$.
	Then the following position $(k, \AA, \BB)$ is a valid starting position for a game $\FS^\opset_{k}(\AA, \BB)$.
	We can obtain a winning strategy for S in this new game using the induction hypothesis.
	We finally obtain the full winning strategy for S by combining the first move described below to the strategy given by the induction hypothesis.
	We divide the proof into cases according to the outermost connective of $\phi$.
	We again handle only one of each pair of dual cases.
	\begin{itemize}
		\item $\phi$ is a literal: We know that $\phi$ separates $\AA_0$ from $\BB_0$ so S wins by making a literal move choosing $\phi$.

		\item $\phi = \psi_1 \bor \psi_2$: S chooses $\AA_i = \{ A \in \AA_0 \mid A \vDash \psi_i\}$ for $i \in \{1,2\}$, $k_1 = \size(\psi_1)$ and $k_2 = k-k_1$.
		Since $\phi$ separates $\AA_0$ from $\BB_0$, we have $\AA_0 \vDash \phi$ so $\AA_1 \cup \AA_2 = \AA$.
		On the other side, $\BB_0 \vDash \sneg\phi$ so $B \nvDash \psi_1$ and $B \nvDash \psi_2$ for every $B \in \BB_0$.
		Now, no matter which number $i \in \{1,2\}$ D chooses, in the following position $(k_i, \AA_i, \BB_0)$, the formula $\psi_i$ will separate $\AA_i$ from $\BB_0$.
		In addition, $k_1 \leq \size(\psi_1)$ and $k_2 = k_0 - k_1 \leq \size(\phi) - \size(\psi_1) = \size(\psi_2)$.
		By induction hypothesis S has a winning strategy for both games $\FS^\opset_{k_i}(\AA_i, \BB_0)$.

		\item  $\phi = \psi_1 \laor \psi_2$: Again we have $\AA_0 \vDash \phi$ so for every $A \in \AA_0$, there is a split $(A_1, A_2)$ such that $A_1 \vDash \psi_1$ and $A_2 \vDash \psi_2$.
		S chooses such a split for every $A \in \AA_0$.
		On the other side, $\BB_0 \vDash \sneg \phi$ so for every $B \in \BB_0$ and every split $(B_1, B_2)$ we have $B_1 \nvDash \psi_1$ or $B_2 \nvDash \psi_2$.
		For each $B \in \BB_0$, S chooses $f_B$ so that if $f_B(B_1, B_2) = i$, then $B_i \nvDash \psi_i$.
		Now, no matter which number $i \in \{1,2\}$ D chooses, in the following position $(k_i, \AA_i, \BB_i)$, the formula $\psi_i$ will separate $\AA_i$ from $\BB_i$.
		S deals with the resource $k_0$ just like in the previous case.
		By induction hypothesis S has a winning strategy for both games $\FS^\opset_{k_i}(\AA_i, \BB_i)$.

		\item $\phi = \psi_1 \sor \psi_2$: Same as the $\laor$-case except all splits considered are strict.
		\qedhere
	\end{itemize}
\end{proof}

\noindent
Before we move on to the lower bounds, we prove a very standard lemma for formula size games stating that if at any time the same team ends up on both sides of the game, D wins.

\begin{lem}\label{sameteam}
	If in a position $P = (k, \AA, \BB)$ there is a team $T \in \AA \cap \BB$, D has a winning strategy from position $P$.
\end{lem}
\begin{proof}
	As long as there is $T \in \AA \cap \BB$, if S makes a literal move,
	D wins.
	We show that D can maintain this condition. We again omit the cases of dual operators.
	\begin{itemize}
		\item $\bor$-move: S chooses sets $\AA_1, \AA_2 \subseteq \AA$.
		Since $\AA_1 \cup \AA_2 = \AA$, we have $T \in \AA_i$ for some $i \in \{1,2\}$.
		Then D chooses the following position $(k_i, \AA_i, \BB)$ and we have $T \in \AA_i \cap \BB$.
		\item $\laor$-move: Let $(T_1, T_2)$ be the split S chooses for $T$ on the left side.
		On the right side S must choose $i = f_T(T_1, T_2) \in \{1,2\}$.
		Then D chooses the following position $(k_i, \AA_i, \BB_i)$ and we have $T_i \in \AA_i \cap \BB_i$.
		\item $\sor$-move: Same as the $\laor$-move except the split must be strict.
	\end{itemize}
	Since S must eventually make a literal move,
	D wins the game.
\end{proof}

\subsection{Lower bounds via the formula size game}

In this section we use the formula size game to show lower bounds for the lengths of formulas defining atoms of dependency in the positive fragment of propositional team logic.
We first state all of the bounds as a theorem and prove them in the rest of the section.

For natural numbers $k$ and $m$, ${[k]}_m$ is the remainder of $k$ modulo $m$.

\begin{thm}\label{thm:lower-bounds}
	Let $\Sigma = \{\bor, \band, \laor, \sor\}$, $n,m \geq 1$, and $\Phi_n = \{p_1, \dots, p_n\}$.
	\begin{enumerate}
		\item If $m \leq 2^n$ and $k < m$, then a $\PL(\Sigma)$-formula, that defines the property $|T| \equiv k \mod m$ of $\Phi_n$-teams $T$, has length at least $2^n-{[2^n - k]}_m$. In particular, a formula that defines even parity has length at least $2^n$.
		\item A $\PL(\Sigma)$-formula that defines cardinality $k \leq 2^n$ of $\Phi_n$-teams has length at least $k$.
		\item A $\PL(\Sigma)$-formula that defines $p_1 \cdots p_n \subseteq q_1 \cdots q_n$ has length at least $2^n$.
		\item A $\PL(\Sigma)$-formula that defines $p_1 \cdots p_n \perp q_1 \cdots q_m$ has length at least $2^{n+m}$.
		\item A $\PL(\Sigma)$-formula that defines $p_1 \cdots p_n \anon q$ has length at least $2^{n+1}$.
	\end{enumerate}
\end{thm}

\noindent
Note that for $\anon$ we only consider a single argument on the right-hand side.
While this is an exponential lower bound (in $n$), a tight bound in both $n$ and $m$ (cf.\ Table~\ref{tab:dep-atoms}) is still open.

Our approach to proving these bounds is similar to that of Hella and Väänänen in~\cite{HellaV15}.
They used a formula size game for propositional logic to show that defining the parity of the number of ones in a propositional assignment of length $n$ requires a formula of length $n^2$.
We focus on teams that differ only by one assignment and define a measure named \emph{density} as in~\cite{HellaV15}, although our definition is slightly different.

\begin{defi}
	Let $T$ be a team.
	A team $T'$ is a \emph{neighbour} of $T$, if $T' = T \setminus \{s\}$ for some assignment $s \in T$.

	\noindent Let $\AA$ be a set of teams.
	The number of neighbours of $T$ in the set $\AA$ is denoted by $N(T, \AA)$,
	\[
	N(T, \AA) \dfn |\{ A \in \AA \mid A \text{ is a neighbour of } T\}|.
	\]
	The \emph{density} of the pair $(\AA, \BB)$ is
	\[
	D(\AA, \BB) \dfn \max\{N(A, \BB) \mid A \in \AA\}.
	\]
\end{defi}

\noindent
We shall use density as an invariant for the formula size game.
Essentially we will show that a certain number of the resource $k$ must be expended before a literal move can be made.
First we show that S cannot make a successful literal move when density is too high.

\begin{lem}\label{atom}
	If $D(\AA, \BB) > 1$, then no literal separates $\AA$ from $\BB$.
\end{lem}
\begin{proof}
	If $D(\AA, \BB) > 1$, at least one team $A \in \AA$ has two neighbours $B_1, B_2 \in \BB$.
	Now any positive literal $l$ (with respect to $\sneg$) true in $A$ is also true in $B_1$ and $B_2$ since they are subteams of $A$.
	On the other hand, since $B_1$ and $B_2$ are different neighbours of $A$, we have $B_1 \cup B_2 = A$.
	For a negative literal $\sneg l$, assume that $B_1 \nvDash \sneg l$ and $B_2 \nvDash \sneg l$.
	This means that $B_1 \vDash l$ and $B_2 \vDash l$, so by union closure, $A \vDash l$ and consequently $A \nvDash \sneg l$.
\end{proof}

We proceed to show that density behaves well with respect to the moves of the game.

\begin{lem}\label{lem:dens}
	Let $(k, \AA, \BB)$ be a position in a game $\FS^\opset_{k_0}(\AA_0, \BB_0)$.
	\begin{enumerate}
		\item If S makes a $\bor$-move, and the possible following positions are $(k_1, \AA_1, \BB)$ and $(k_2, \AA_2, \BB)$, then	$D(\AA_1, \BB) + D(\AA_2, \BB) \geq D(\AA, \BB)$.
		\item If S makes a $\band$-move, and the possible following positions are $(k_1, \AA, \BB_1)$ and $(k_2, \AA, \BB_2)$, then $D(\AA, \BB_1) + D(\AA, \BB_2) \geq D(\AA, \BB)$.
		\item If S makes a $\laor$-move or $\sor$-move, and the possible following positions are $(k_1, \AA_1, \BB_1)$ and $(k_2, \AA_2, \BB_2)$, then	$D(\AA_1, \BB_1) + D(\AA_2, \BB_2) \geq D(\AA, \BB)$ or D has a winning strategy from one of the following positions.
	\end{enumerate}
\end{lem}
\begin{proof}
	Let $A$ be one of the teams in $\AA$ with most neighbours in $\BB$.
	\begin{enumerate}
		\item Since $\AA_1 \cup \AA_2 = \AA$, we may assume by symmetry that $A \in \AA_1$.
		Since all the same neighbours of $A$ are still in $\BB$, we get $D(\AA_1, \BB) + D(\AA_2, \BB) \geq D(\AA_1, \BB) \geq D(\AA, \BB)$.
		\item Since $\BB_1 \cup \BB_2 = \BB$, the neighbours of $A$ are split between $\BB_1$ and $\BB_2$ so \\ $D(\AA, \BB_1) + D(\AA, \BB_2) \geq N(A, \BB_1) + N(A, \BB_2) \geq N(A, \BB) = D(\AA, \BB)$.
		\item Let $(A_1, A_2)$ be the (strict) split of $A$ chosen by S.
		Suppose $B = A \setminus \{a\}$ is a neighbour of $A$ in $\BB$.
		Then $(B_1,B_2) \dfn (A_1 \setminus \{a\}, A_2 \setminus \{a\})$ is a split of $B$, and is strict if $(A_1, A_2)$ is strict.
		Let $f_B : \mathrm{(S)}\splits(B) \to \{1,2\}$ be the function chosen by S for the team $B$ and $i \dfn f_B(B_1, B_2)$.
		If $a \notin A_i$, then $A_i = B_i \in \AA_i \cap \BB_i$ and by Lemma~\ref{sameteam}, D has a winning strategy from the position $(k_i, \AA_i, \BB_i)$.
		Consequently, we proceed with the case where $a \in A_i$ for all $A,B$ as above.
		Then $B_i = A_i \setminus \{a\}$ is a neighbour of $A_i$ in $\BB_i$, \ie, on the opposite side in the position $(k_i, \AA_i, \BB_i)$.
		We see that for each neighbour $B$ of $A$, we obtain a neighbour of $A_1$ in $\BB_1$, or one of $A_2$ in $\BB_2$.
		Furthermore, if $B = A \setminus \{a\}$ and $B' = A \setminus \{a'\}$ are distinct neighbours of $A$, then $A_i\setminus \{a\}$ and $A_i\setminus \{a'\}$ are distinct neighbours of $A_i$.
		For this reason, $D(\AA_1, \BB_1) + D(\AA_2, \BB_2) \geq D(\AA, \BB)$. \qedhere
	\end{enumerate}
\end{proof}

\noindent
For the rest of this section, we study a fragment $\PL(\opset)$ with operators from $\opset = \{\bor, \band, \laor, \sor\}$. All results are lower bounds for this fragment and are naturally preserved by any fragment $\PL(\opset')$ with $\opset' \subseteq \opset$.

We gather the above lemmas as the following theorem stating the usefulness of density.

\begin{thm}\label{winstrat}
	If $k_0 < D(\AA_0, \BB_0)$, then D has a winning strategy in the game $\FS^\opset_{k_0}(\AA_0, \BB_0)$.
\end{thm}
\begin{proof}
	We define a strategy $\delta$ for D and show that if D plays according to $\delta$, the condition $k < D(\AA, \BB)$ is maintained in all positions $(k, \AA, \BB)$.

	Let $(k, \AA, \BB)$ be a position of the game $\FS^\opset_{k_0}(\AA_0, \BB_0)$.
	By induction hypothesis, $k < D(\AA, \BB)$.
	\begin{itemize}
		\item If S makes a $\bor$-move, then by the first item of Lemma~\ref{lem:dens}, $D(\AA_1, \BB) + D(\AA_2, \BB) \geq D(\AA, \BB)$.
		Assume for contradiction that $k_i \geq D(\AA_i, \BB)$ for $i \in \{1,2\}$.
		Then
		\[
		k = k_1 + k_2 \geq D(\AA_1, \BB) + D(\AA_2, \BB) \geq D(\AA, \BB) > k,
		\]
		which is a contradiction.
		Therefore $k_i < D(\AA_i, \BB)$ for some $i \in \{1,2\}$ and D chooses that $i$ to continue the game.
		\item The case of a $\band$-move is similar, the second item of Lemma~\ref{lem:dens}.
		\item If S makes a $\laor$-move, then by the third item of Lemma~\ref{lem:dens}, D has a winning strategy from a following position $(k_i, \AA_i, \BB_i)$ or $D(\AA_1, \BB_1) + D(\AA_2, \BB_2) \geq D(\AA, \BB)$.
		In the first case D chooses the position $(k_i, \AA_i, \BB_i)$ and follows the strategy given by the lemma.
		In the second case D chooses a following position that maintains the condition $k < D(\AA, \BB)$ just like in the $\bor$-case above.
		\item If S makes a literal move,
		since $D(\AA, \BB) > k \geq 1$, by Lemma~\ref{atom}, D wins the game. Note that the case $k = 0$ is not possible since all binary connective moves lead to positions with positive $k$, and a literal move always ends the game. \qedhere
	\end{itemize}
\end{proof}

\begin{lem}\label{neighbours}
	No set $\AA$ of $\Phi$-teams can be defined with a $\PL(\opset)$-formula of width less than $D(\AA, \teams(\Phi) \setminus \AA)$.
\end{lem}
\begin{proof}
	Let $\BB := \teams(\Phi) \setminus \AA$.
	Now defining $\AA$ amounts to separating $\AA$ from $\BB$.
	If $k < D(\AA, \BB)$, then by Theorem~\ref{winstrat}, D has a winning strategy in the game $\FS^\opset_{k_0}(\AA, \BB)$ and by Theorem~\ref{gametheorem}, $\AA$ and $\BB$ cannot be separated by a formula with width $k$.
\end{proof}

With the above lemma, we are now in the position to prove the main theorem of this section.

\begin{proof}[Proof of Theorem~\ref{thm:lower-bounds}]
	We find in each case a team which satisfies the desired property $\AA$ and has the desired number of neighbours which do not.
	Then $D(\AA, \teams(\Phi) \setminus \AA)$ is greater than or equal to the desired number and the claim follows from Lemma~\ref{neighbours} along with the fact that length is always greater than width.
	\begin{enumerate}
		\item First is the cardinality $k \mod m$ of $\Phi_n$-teams.
		Let $k' = 2^n - {[2^n-k]}_m$.
		We first note that $k' \leq 2^n$ so there is a $\Phi_n$-team $T_1$ with cardinality $k'$.
		Furthermore,
		\[
		k' \equiv {[2^n - 2^n + k]}_m \equiv k \mod m.
		\]
		Now $|T_1| \equiv k \mod m$ and $T_1$ has $k'$ neighbours with smaller cardinality.
		\item For a specific cardinality $k \leq 2^n$, if $T_2$ is any team with cardinality $k$, then $T_2$ clearly has $k$ neighbours with a smaller cardinality.
		\item Next is the inclusion atom $p_1 \cdots p_n \subseteq q_1 \cdots q_n$.
		If $s(p_1)\cdots s(p_n)$ is a binary representation of the number $i$, we denote this by $s(\vv{p}) = i$.
		For $i \in \{0, \dots, 2^n -1\}$, let $s_i$ be the assignment with $s_i(\vv{p}) = i$ and $s_i(\vv{q}) = {[i + 1]}_{2^n}$.
		Let $T_3 := \{s_i \mid i \in \{0, \dots, 2^n-1\}\}$.
		Now $\vv{p}$ and $\vv{q}$ both get all possible values so $T_3 \vDash p_1\cdots p_n \subseteq q_1\cdots q_n$.
		Furthermore, for any $s_i \in T_3$, we have $T_3 \setminus \{s_i\} \nvDash p_1\cdots p_n \subseteq q_1\cdots q_n$ since $\vv{p}$ gets the value ${[i+1]}_{2^n}$ but $\vv{q}$ does not.
		Thus there are $|T_3| = 2^n$ neighbours of $T_3$ which do not satisfy the inclusion atom.
		\item For the independence atom $p_1 \cdots p_n \perp q_1 \cdots q_m$, let $T_4$ be the full team with domain $\{p_1, \dots, p_n, q_1, \dots, q_m\}$.
		Clearly $T_4 \vDash p_1\cdots p_n \perp q_1 \cdots q_m$ and $(T_4 \setminus \{s\}) \nvDash p_1\cdots p_n \perp q_1 \cdots q_m$ for any assignment $s \in T_4$.
		Thus there are $|T_4| = 2^{n+m}$ neighbours of $T_4$ which do not satisfy the independence atom.
		\item Finally, for the anonymity atom $p_1 \cdots p_n \anon q$, let $T_5$ be the full team with domain $\{p_1, \dots, p_n, q\}$.
		Clearly $T_5 \vDash p_1 \cdots p_n \anon q$ and $T_5 \setminus \{s\} \nvDash p_1 \cdots p_n \anon q$ for any $s \in T_5$.
		We now have $|T_5| = 2^{n+1}$ neighbours of $T_5$ which do not satisfy the anonymity atom. \qedhere
	\end{enumerate}
\end{proof}
\noindent In the above, we did not prove lower bounds for the atoms of dependence and exclusion.
The reason for this is that the invariant we use for the formula size game is density, which is defined via the neighbourship relation.
The remaining two atoms are downward closed, so a team which satisfies such an atom cannot have any neighbours which do not satisfy the same atom.
For this reason, the above strategy fails for these two atoms.
We present a different approach in the next section.

\subsection{Lower bounds via upper dimension}

For the lower bounds of dependence and exclusion atoms we employ the notion of \emph{upper dimension}, which was successfully used to prove lower bounds by Hella et al.~\cite{HellaLSV14}.
Their paper mainly concerns the expressive power of modal dependence logic, but at the end it is shown that defining the dependence atom in modal logic with Boolean disjunction $\bor$ requires a formula with length at least $2^n$.
However, the logic they consider again has downward closure.
The existential fragment is not downward closed, so we adapt the technique of Hella et al.\ accordingly.
We first state the lower bounds as a theorem and then prove it in this section.

\begin{thm}\label{thm:dimension-lower-bounds}
 Let $\opset = \{\bor, \land, \lor, \sor\}$ and $n \geq 1$.
 \begin{itemize}
  \item A $\PL(\opset)$-formula that defines $\dep{p_1 \cdots p_n; q}$ has length at least $2^n$.
  \item A $\PL(\opset)$-formula that defines $p_1 \cdots p_n \mid q_1 \cdots q_n$ has length at least $2^n$.
 \end{itemize}
\end{thm}
\noindent For now, we will assume that $\opset = \{\bor,\land,\lor\}$.
We will show in the next subsection that this imposes no restriction on the results, as for every $\PL(\{\bor,\land,\lor,\sor\})$-formula that is local there is an equivalent $\PL(\{\bor,\land,\lor\})$-formula of the same size.

\begin{defi}
 Let $\phi \in \PL(\opset, \Phi)$.
 A \emph{generator} of $\phi$ is a set $\GG(\phi)$ of pairs $(S,U)$ such that $S \subseteq U$, and for each $\Phi$-team $T$ it holds that $T \vDash \phi$ precisely if there is $(S,U) \in \GG(\phi)$ such that $S \subseteq T \subseteq U$.
 The \emph{upper dimension} $\Dim(\GG)$ of $\GG$ is the number of distinct upper bounds in $\GG$:
 \[
  \Dim(\GG) \dfn |\{ U : (S,U) \in \GG \}|\text{.}
 \]
 The \emph{upper dimension} of $\phi$, denoted $\Dim(\phi)$, is the minimal upper dimension of a generator of $\phi$:
 \[
  \Dim(\phi) \dfn \min\{\Dim(\GG) \mid \GG \text{ is a generator of }\phi\}\text{.}
 \]
\end{defi}

That we count only the upper bounds $U$ is analogous to Hella et al.~\cite{HellaLSV14}, who considered downward closed formulas and defined generators only in terms of $U$.
Indeed, with downward closure we could simply set $S \dfn \emptyset$ and obtain a definition equivalent to theirs.
For arbitrary formulas $\phi$ however (even with the empty team property), we could have $(S,U) \in \GG(\phi)$, but $\emptyset \subsetneq X \subsetneq S \subseteq U$ for some $X$ such that $X \nvDash \phi$.
Since the subformulas defining a downward closed formula are not necessarily downward closed, the inductive proofs in our results only work if we additionally keep track of the $S$.

\begin{lem}
 Let $\phi, \psi \in \PL(\opset, \Phi)$ and $\Phi = \{p_1,\ldots,p_n\}$. We have the following estimates:
 \begin{itemize}
  \item $\Dim(l) \leq 1$ for any $\Phi$-literal $l$,
  \item $\Dim(\phi \land \psi) \leq \Dim(\phi)\cdot\Dim(\psi)$,
  \item $\Dim(\phi \lor \psi) \leq \Dim(\phi)\cdot\Dim(\psi)$,
  \item $\Dim(\phi \bor \psi) \leq \Dim(\phi) + \Dim(\psi)$,
 \end{itemize}
\end{lem}
\begin{proof}
 For the binary connectives, let $\GG(\phi)$ and $\GG(\psi)$ be minimal generators of $\phi$ and $\psi$, respectively.
 \begin{itemize}
  \item
        Let $T$ be the full $\Phi$-team.
        Any positive literal $l \in \{ p, \neg p, \top, \bot \mid p \in \Phi\}$ has flatness, so $\{(\emptyset,\{ s \in T \mid \{s\} \vDash l\})\}$ generates $l$.
        The negative literals $l \in \{ \sneg p,\sneg \neg p, \sneg \bot \mid p \in \Phi \}$ are upward closed, so $\{ (\{s\},T) \mid s \in T : \{s\} \vDash l\}$ generates $l$.
        Finally, $\sneg \top$ is unsatisfiable, so it has the empty generator.

  \item	For the conjunction, it is easy to check that $\GG(\cap) \dfn \{ (S_1 \cup S_2, U_1 \cap U_2) \mid (S_1,U_1) \in \GG(\phi), (S_2,U_2) \in \GG(\psi)\}$ is a generator of $\phi \land \psi$, so
        \[
         \Dim(\phi \land \psi) \leq \Dim(\GG(\cap)) \leq \Dim(\GG(\phi)) \cdot \Dim(\GG(\psi)) = \Dim(\phi)\cdot\Dim(\psi).
        \]

  \item 	For the lax disjunction, let $\GG(\cup) \dfn \{(S_1 \cup S_2,U_1 \cup U_2) \mid (S_1,U_1) \in \GG(\phi), (S_2, U_2) \in \GG(\psi)\}$.
        If $T \vDash \phi \lor \psi$ via some split $(T_1, T_2)$, there are $(S_1,U_1) \in \GG(\phi)$ and $(S_2,U_2) \in \GG(\psi)$ such that $S_i \subseteq T_i \subseteq U_i$ for $i \in \{1,2\}$.
        Then $S_1 \cup S_2 \subseteq T \subseteq U_1 \cup U_2$.

        Conversely, assume $(S_1,U_1) \in \GG(\phi)$ and $(S_2,U_2) \in \GG(\psi)$ such that $S_1 \cup S_2 \subseteq T \subseteq U_1 \cup U_2$.
        Define $T_i \dfn (T \cap U_i) \cup S_i$.
        Then $(T_1,T_2)$ is a split of $T$, and $S_i \subseteq T_i \subseteq U_i$ (\wloss $S_i \subseteq U_i$).
        Consequently, $T_1 \vDash \phi$ and $T_2 \vDash \psi$, so $T \vDash \phi \lor \psi$.
        Thus $\GG(\cup)$ is a generator of $\phi \lor \psi$ and
        \[
         \Dim(\phi \lor \psi) \leq \Dim(\GG(\cup)) \leq \Dim(\GG(\phi)) \cdot \Dim(\GG(\psi)) = \Dim(\phi)\cdot\Dim(\psi).
        \]

  \item 	For the Boolean disjunction, clearly $\GG(\phi) \cup \GG(\psi)$ is a generator of $\phi \bor \psi$.\qedhere
 \end{itemize}
\end{proof}

\noindent
Let $\occ(\phi)$ denote the number of occurrences of $\bor$ inside $\phi$.

\begin{lemC}[{\cite[Proposition 5.9]{HellaLSV14}}]\label{lem:dimension-of-pl}
 Let $\phi \in \PL(\Sigma)$.
 Then $\Dim(\phi) \leq 2^{\occ(\phi)}$.
\end{lemC}
\begin{proof}
 By induction on $\phi$, using the previous lemma.
 For literals $\phi = l$, $\Dim(l) \leq 1 = 2^{\occ(l)}$.
 For $\nabla \in \{\land,\lor\}$, it holds that
 \begin{align*}
  \Dim(\psi \nabla \theta) \leq \; & \Dim(\psi) \cdot \Dim(\theta)                                                         \\
  \leq \;                          & 2^{\occ(\psi)} \cdot 2^{\occ(\theta)} = 2^{\occ(\psi)+\occ(\theta)} =  2^{\occ(\phi)}
 \end{align*}
 and for the Boolean disjunction,
 \begin{align*}
  \Dim(\psi \bor \theta) \leq \; & \Dim(\psi) + \Dim(\theta) \leq \Dim(\psi)\cdot \Dim(\theta)+1                                                  \\
  \leq \;                        & 2^{\occ(\psi)} \cdot 2^{\occ(\theta)} + 1 \leq 2^{\occ(\psi)+\occ(\theta) +1} = 2^{\occ(\phi)}\text{.}\qedhere
 \end{align*}
\end{proof}

Next, we show that the upper dimension of the dependence atom and the exclusion atom is at least doubly exponential.

\begin{lem}\label{lem:dimension-of-atoms}
 Let $n \geq 1$, let $p_1,\ldots,p_n,q_1,\ldots,q_n \in \Phi$ be pairwise distinct propositions, $\vv{p} = (p_1,\ldots,p_n)$, and $\vv{q} = (q_1,\ldots,q_n)$.
 Then $\Dim(\dep{\vv{p};q_1}) \geq 2^{2^n}$ and $\Dim(\vv{p} \mid \vv{q}) \geq 2^{2^n} - 2$.
\end{lem}
\begin{proof}
We prove a more general result and then apply it to the two atoms.
Let $\phi$ be any formula and $\Phi = \Prop(\phi)$.
We show that the size of a generator of $\phi$ is always at least the number of maximal $\Phi$-teams of $\phi$, where a $\Phi$-team $X$ is \emph{maximal} if it satisfies $\phi$ but no $\Phi$-team $Y$ with $Y \supsetneq X$ satisfies $\phi$.
Suppose that $\phi$ has $m$ distinct maximal teams, but $\GG$ is a generator of $\phi$ with $|\GG| < m$.
Then there are distinct maximal teams $X_1,X_2$ and pairs $(S_1,U), (S_2,U) \in \GG$ such that $X_1,X_2 \subseteq U$.
Since $X_1$ is maximal and $U \vDash \phi$ by definition of generator, we have $X_1 = U$.
But by the same argument $X_2 = U$, contradiction.

\medskip

Next, we show that the atoms have at least $2^{2^n}$ maximal teams.
We start with the dependence atom.
For each $f \colon {\{0,1\}}^n \to \{1,0\}$, let
\[
 X(f) \dfn \{ s \colon \Phi \to \{0,1\} \mid f(s(\vv{p})) = s(q_1) \}\text{.}
\]
Then $X(f)$ is maximal for $\dep{\vv{p};q_1}$.
Since $f_1 \neq f_2$ implies $X(f_1) \neq X(f_2)$, there are at least $2^{2^n}$ distinct maximal teams.

\medskip

For the exclusion atom, we proceed similarly.
A function $f \colon {\{0,1\}}^n \to \{1,0\}$ is \emph{non-constant} if $f(\vv{b}) \neq f(\vv{b'})$ for some $\vv{b},\vv{b'} \in {\{0,1\}}^n$.
Now, for all non-constant $f \colon {\{0,1\}}^n \to \{1,0\}$, let
\[
 X(f) \dfn \{ s \colon \Phi \to \{0,1\} \mid f(s(\vv{p})) = 1 \text{ and }  f(s(\vv{q})) = 0 \}\text{.}
\]
Clearly $X(f) \vDash \vv{p} \mid \vv{q}$, as for every $s \in X(f)$ we have $f(s(\vv{p})) \neq f(s(\vv{q}))$, hence $s(\vv{p}) \neq s(\vv{q})$.

Next, we show that these are distinct teams, \ie, $f_1 \neq f_2$ implies $X(f_1) \neq X(f_2)$.
Suppose $f_1 \neq f_2$, \wloss there is $\vv{b} \in {\{0,1\}}^n$ such that $f_1(\vv{b}) = 1$ and $f_2(\vv{b}) = 0$.
Consider the assignment $s$ defined by $s(\vv{p}) = \vv{b}$ and $s(\vv{q})$ defined in a way such that $f_1(s(\vv{q})) = 0$ (recall that $f_1$ is non-constant).
Then $s \in X(f_1)$ but $s \notin X(f_2)$.
Consequently, there are $2^{2^n}-2$ such teams (as there are $2^{2^n}-2$ non-constant functions).

It remains to show that these teams are maximal, \ie, $X(f) \subsetneq Y$ implies $Y \nvDash \vv{p} \mid \vv{q}$ for all $\Phi$-teams $Y$.
Suppose $s \in Y \setminus X(f)$.
Then $f(s(\vv{p})) = 0$ or $f(s(\vv{q})) = 1$.
By symmetry, we consider only the first case.
As $f$ is non-constant, there exists $\vv{b} \in {\{0,1\}}^n$ with $f(\vv{b}) = 1$.
Now, define an assignment $s'$ such that $s'(\vv{p}) = \vv{b}$ and $s'(\vv{q}) = s(\vv{p})$.
Then $f(s'(\vv{p})) = 1$ and $f(s'(\vv{q})) = 0$, so $s'\in X(f) \subseteq Y$.
Hence $s,s' \in Y$, but $s'(\vv{q}) = s(\vv{p})$, so $Y \nvDash \vv{p} \mid \vv{q}$.
\end{proof}

We conclude the section with the following exponential lower bounds.

\begin{proof}[Proof of Theorem~\ref{thm:dimension-lower-bounds}]
 We consider the exclusion atom, the dependence atom works analogously.
 Suppose that $\phi \in \PL(\Sigma)$ is equivalent to $p_1\cdots{}p_n \mid q_1\cdots{}q_n$.
 Then by Lemma~\ref{lem:dimension-of-atoms}, $\Dim(\phi) \geq 2^{2^n}-2$, as the upper dimension is a purely semantical property.
 However, by Lemma~\ref{lem:dimension-of-pl}, $\Dim(\phi) \leq 2^{\occ(\phi)} \leq 2^{|\phi|} - 2$.
 With $n \geq 1$, the resulting inequality $2^{2^n} - 2\leq 2^{|\phi|} - 2$ implies $|\phi| \geq 2^n$.
\end{proof}

\subsection{From lax to strict lower bounds}

Before, we proved lower bounds for the dependence and exclusion atom for the for the restricted operator set $\opset = \{\bor,\land,\lor\}$, in particular with only lax disjunction.
Next, we incorporate the strict disjunction $\sor$.

The idea is the following:
Define the \emph{relaxation} $\lax\phi$ of a formula $\phi$ as the formula where every occurrence of $\sor$ is replaced by $\lor$.
We will prove that a formula $\phi$ and its relaxation are equivalent, provided $\phi$ is local.
This is the case in particular for the dependence and the exclusion atom, for which all lower bounds with $\lor$ then also hold with $\sor$.
This additional assumption of locality is needed, since formulas containing $\sor$ can be non-local.
For example, $\nempty \sor \nempty$ is not equivalent to its relaxation $\nempty \lor \nempty \equiv \nempty$.

\smallskip

The intuition is that if $\lax\phi$ is satisfiable, then $\phi$ is also satisfiable if we just make the domain larger, since the only way $\phi$ could be false while $\lax\phi$ is true is that we "run out of assignments" for $\sor$.
But if now $\phi$ is local, then enlarging the domain should have no effect so that then we have $\phi \equiv \lax\phi$.

\smallskip

We begin with proving the first part formally.
If $T$ is a $\Phi$-team and $\Psi \supseteq \Phi$, then the $\Psi$-expansion of $T$ is
\[
T[\Psi] \dfn \Set{ s \colon \Psi \to \{0,1\} | \restr{s}{\Phi} \in T }\text{.}
\]
Intuitively it is obtained from $T$ by duplicating all assignments in $T$ for all possible values for propositions $p \in \Psi \setminus \Phi$.
Observe that $\restr{T[\Psi]}{\Phi} = T$.

\begin{lem}\label{lem:lax-satisfiability}
 Let $\phi \in \PL(\{\land,\lor,\sor\})$.
 If a $\Phi$-team $T$ satisfies $\lax\phi$, then there is a domain $\Psi \supseteq \Phi$ such that $T[\Psi]$ satisfies $\phi$.
\end{lem}
\begin{proof}
 The idea is that any lax splitting can be simulated by a strict splitting by duplicating assignments in the team such that no assignment needs to be used in both halves of the splitting.
 We show the following stronger statement by induction on $\phi$:
 If $\lax\phi$ is satisfied by a $\Phi$-team $T$, then there is a domain $\Psi \supseteq \Phi$ such that, for all domains $\Psi' \supseteq \Psi$ and $\Psi'$-teams $X$, it holds that $\restr{X}{\Psi} = T[\Psi]$ implies $X \vDash \phi$.

 The case where $\phi$ is a literal or a conjunction is straightforward.
 So suppose $\phi = \psi_1 \sor \psi_2$ or $\phi = \psi_1 \lor \psi_2$, and assume $T \vDash \lax\phi = \lax{\psi_1} \lor \lax{\psi_2}$ via a (lax) split $(S_1,S_2)$ of $T$, \ie, $S_i \vDash \lax{\psi_i}$.
 For $i \in \{1,2\}$, there is $\Psi_i \supseteq \Phi$ such that for all $\Psi'_i \supseteq \Psi_i$ and $\Psi'_i$-teams $X_i$ it holds that $\restr{X_i}{\Psi_i} = S_i[\Psi_i]$ implies $X_i \vDash \psi_i$.
 We pick $p \in \Prop \setminus (\Psi_1 \cup \Psi_2)$, and let $\Psi \dfn \Psi_1 \cup \Psi_2 \cup \{p\}$.
 Now assume $\restr{X}{\Psi} = T[\Psi]$ for some $\Psi'$-team $X$, where $\Psi' \supseteq \Psi$.
 We have to show that $X \vDash \psi_1 \sor \psi_2$.
 This holds via the strict split $(Y_1 \cup Z_1,Y_2 \cup Z_2)$ of $X$, where
 \begin{align*}
  Y_1 & \dfn \{ s \in X \mid \restr{s}{\Phi} \in S_1\cap S_2 \text{ and } s(p) = 1\} \\
  Y_2 & \dfn \{ s \in X \mid \restr{s}{\Phi} \in S_1\cap S_2 \text{ and } s(p) = 0\} \\
  Z_1 & \dfn \{ s \in X \mid \restr{s}{\Phi} \in S_1 \setminus S_2 \}                \\
  Z_2 & \dfn \{ s \in X \mid \restr{s}{\Phi} \in S_2 \setminus S_1 \}\tag*{\qedhere}
 \end{align*}
\end{proof}

\noindent
We now prove the second part.

\begin{thm}
 A formula $\phi \in \PL(\{\bor, \land, \lor, \sor\})$ is local if and only if it is equivalent to its relaxation $\lax{\phi}$.
\end{thm}
\begin{proof}
  If $\phi$ is equivalent to $\lax\phi$, then $\phi$ is local by Proposition~\ref{prop:locality}.
  For the converse, let $\phi$ be local.
  We have to prove $\phi \equiv \lax\phi$.
 The direction $\phi \vDash \lax{\phi}$ is easy to prove by induction.
 For the other direction, $\lax{\phi} \vDash \phi$, we first transform $\phi$ and $\lax{\phi}$ into a disjunction of $\PL(\{ \land,\lor,\sor\})$-formulas
 using the distributive laws
 \begin{align*}
  \theta_1 \circ (\theta_2 \bor \theta_3) & \equiv (\theta_1\circ \theta_2) \bor (\theta_1 \circ \theta_3)  \\
  (\theta_1 \bor \theta_2) \circ \theta_3 & \equiv (\theta_1\circ  \theta_3) \bor (\theta_2 \circ \theta_3)
 \end{align*}
 for $\circ \in \{ \land,\lor,\sor\}$.
 We obtain $\phi \equiv \bigbor_{i = 1}^n \psi_i$ and $\lax{\phi} \equiv \bigbor_{i=1}^n \lax{\psi_i}$
 for suitable $\psi_1,\ldots,\psi_n \in \PL(\{ \land,\lor,\sor\})$.\footnote{Such normal forms with $\ovee$ are standard in team logic (cf.\ \cite[Theorem~3.5]{HellaLSV14}, \cite[Theorem~3.4]{Kon15}, \cite[Lemma~4.9]{YangV16}, \cite[Proposition~6.2]{Virtema17}).} 

 \smallskip

 Let now $T$ be a $\Phi$-team such that $T \vDash \lax\phi$, where $\Phi \supseteq \Prop(\phi)$.
 Then $T \vDash \lax{\psi_i}$ for some $i$.
 By Lemma~\ref{lem:lax-satisfiability}, there is a domain $\Psi \supseteq \Phi$ such that $T[\Psi] \vDash \psi_i$, which implies that $T[\Psi] \vDash \phi$.
 Since $\restr{T[\Psi]}{\Phi} = T$ and $\phi$ is local, we conclude $T \vDash \phi$, as desired.
\end{proof}

\section{Polynomial upper bounds for team properties}\label{sec:upper}

In this section, we complement the exponential lower bounds presented in Theorem~\ref{thm:lower-bounds} by polynomial upper bounds in the fragment $\PL(\{\bor,\land,\lor\})$.
Notably, among these polynomially definable properties are the \emph{negations} of all atoms of dependency considered previously.
This exhibits an interesting asymmetry of succinctness between the standard atoms of dependency and their negations.
For the parity of teams there is no such asymmetry and we have exponential lower bounds for both even and odd cardinality.
Nevertheless, in the subsequent subsection, we will present a polynomial upper bound for parity in a stronger logic than $\PL(\{\bor,\land,\lor,\sor\})$.

\subsection{Upper bounds for the atoms of dependency}

As with the lower bounds, we will first state the theorem and prove it with a series of lemmas.
The length of a tuple $\vv\phi = (\phi_1,\ldots,\phi_n)$ of formulas is $|\vv\phi| \dfn \sum_{i=1}^n|\phi_i|$.
The \emph{negation} of a formula $\phi$ is $\sneg \phi$.
Throughout this section, let $\vv\alpha, \vv\beta,\vv\gamma$ always denote tuples of purely propositional  formulas $\vv\alpha = (\alpha_1,\ldots,\alpha_n)$, $\vv\beta = (\beta_1,\ldots,\beta_m)$, and $\vv\gamma = (\gamma_1,\ldots,\gamma_k)$, where $n,m,k \geq 0$.

\begin{thm}\label{thm:upper-bounds}
  Let $\Sigma \supseteq \{ \bor, \land, \lor \}$.
  \begin{itemize}
    \item The dependence atom $\dep{\vv\alpha;\vv\beta}$ is equivalent to the negation of a $\PL(\Sigma)$-formula of length $\calO(|\vv\alpha\vv\beta|)$.
    \item The exclusion atom $\vv{\alpha} \mid \vv{\beta}$ is equivalent to the negation of a $\PL(\Sigma)$-formula of length $\calO(n|\vv\alpha\vv\beta|)$.
    \item The inclusion atom $\vv\alpha\subseteq \vv\beta$ is equivalent to the negation of a $\PL(\Sigma)$-formula of length $\calO(n|\vv\alpha\vv\beta|)$.
    \item The conditional independence atom $\vv\alpha\perp_{\vv\gamma}\vv\beta$ is equivalent to the negation of a $\PL(\Sigma)$-formula of length $\calO(n(n+m+k)|\vv\alpha\vv\beta\vv\gamma|)$.
    \item The anonymity atom $\vv\alpha\anon \beta$ is equivalent to the negation of a $\PL(\Sigma)$-formula of length $\calO(n|\beta|+|\vv\alpha|)$.
  \end{itemize}
  Additionally, for the dependence and exclusion atoms, $\Sigma \supseteq \{ \bor, \land, \sor\}$ yields the same result.
  Furthermore, all these formulas are logspace-computable.
\end{thm}
 \begin{proof}
  We prove these results in Lemmas~\ref{lem:dep} to~\ref{lem:anon-unary}.
   For the formulas that are equivalent to the negations of the dependence and exclusion atom, note that every occurrence of $\lor$ in them is of the form $\alpha \lor \varphi$ for purely propositional $\alpha$.
   But then $\alpha \lor \varphi \equiv \alpha \sor \varphi$ by Proposition~\ref{prop:strict-lax}.
   For this reason, these results hold for $\Sigma \supseteq \{ \bor, \land, \sor\}$ as well.
 \end{proof}

\subsection*{Dependence atom.}

It is well-known that the dependence atom can be efficiently rewritten by means of other connectives in most flavors of team logic that have unrestricted negation (see, \eg,~\cite{vaa07,Kon15,HannulaKVV18}).
For the sake of completeness, we will also state such a formula here.

The following formula expresses the negation of the dependence atom $\dep{\vv{\alpha};\vv{\beta}}$ and has length $\calO(|\vv\alpha\vv\beta|)$.
Recall the defined abbreviations $\E \alpha \dfn \top \lor (\nempty \land \alpha)$ and $(\vv\alpha \equi \vv\beta) \dfn \bigwedge_{i=1}^n ((\alpha_i \land \beta_i) \lor (\neg \alpha_i \land \neg \beta_i))$, which we will extensively use in this section.

\noindent{}The following formula defines $\sneg\dep{\vv\alpha;\vv\beta}$.
\begin{align*}
 \varphi(\vv{\alpha};\vv{\beta}) & \dfn \top \lor \bigg(\bigwedge_{i=1}^n (\alpha_i \bor \neg \alpha_i) \land \bigbor_{i=1}^m(\E \beta_i \land \E \neg \beta_i) \bigg)
\end{align*}

\begin{lem}\label{lem:dep}
 $\sneg\dep{\vv{\alpha};\vv{\beta}} \equiv \varphi(\vv{\alpha};\vv{\beta})$.
\end{lem}
\begin{proof}
 Analogously to~\cite[Proposition 2.5]{HannulaKVV18}.
\end{proof}

Next, we require the abbreviation $\alpha \hook \varphi \dfn \neg \alpha \lor (\alpha \land \varphi)$, or equivalently, with strict splitting, $\alpha \hook \varphi \dfn \neg \alpha \sor (\alpha \land \varphi)$.
It was introduced by Galliani~\cite{ga15} and has the semantics
$T \vDash \alpha \hook \varphi \Leftrightarrow T_\alpha \vDash \varphi$,
where $T_\alpha \dfn \{ s \in T \mid s \vDash \alpha \}$.

Before we define the next atom, we introduce two helper formulas $\theta^=$ and $\theta^{\neq}$, which we will explain below.
\begin{align*}
  \theta^=(\vv{\alpha};\vv{\beta};\gamma) &
  \dfn \bigwedge_{i=1}^n \bigbor_{l \in \{\top,\bot\}} \Big( (\gamma \land (\alpha_i \equi l)) \lor (\neg \gamma \land (\beta_i \equi l)) \Big)\\
  \theta^{\neq}(\vv{\alpha};\vv{\beta};\gamma) &
  \dfn \bigvee_{i = 1}^n \Big(\E \gamma \land
  \bigbor_{\mathclap{l \in \{\top,\bot\}}} \big( (\gamma \land (\alpha_i \equi l))
  \lor (\neg \gamma \land (\beta_i \nequi l)) \big) \Big)
\end{align*}
These are $\PL(\{\bor,\land,\lor\})$-formulas of length $\calO(n|\gamma| + |\vv\alpha|+|\vv\beta|)$.

The purpose of $\theta^=(\vv\alpha,\vv\beta,\gamma)$ and $\theta^{\neq}(\vv\alpha,\vv\beta,\gamma)$ is the following.
The definitions of the various dependency atoms are all based on comparison of pairs of assignments in a team.
For instance, $\vv\alpha \mid \vv\beta$ holds if $s(\vv\alpha) \neq s'(\vv\beta)$ for all $s,s' \in T$, and so on.
Loosely speaking, $\theta^=(\vv\alpha,\vv\beta,\gamma)$ and $\theta^{\neq}(\vv\alpha,\vv\beta,\gamma)$ test the values $s(\vv\alpha)$ and $s'(\vv\beta)$ for equality resp.\ inequality for pairs $(s,s') \in T_\gamma \times T_{\neg \gamma}$.
The restriction to $T_\gamma \times T_{\neg \gamma}$ is unfortunately necessary in our implementation of $\theta^=$ and $\theta^{\neq}$, so $s$ and $s'$ must differ in some formula $\gamma$ that is known \emph{a priori}.
While this seems to complicate the matter, we can actually find such $\gamma$ for all of the atoms of dependency.

Before we proceed with defining the atoms, we prove the semantics of $\theta^=$ and $\theta^{\neq}$.
Another constraint is that they work only for the subclass of teams $T$ where $|\{ s(\vv\alpha) \mid s \in T_\gamma\}| = 1$, \ie, all $s \in T_\gamma$ agree on the value $s(\vv\alpha)$, but this again suffices for our purpose.

\begin{lem}\label{lem:compare-forms}
  Let $T$ be a team such that $|\{ s(\vv\alpha) \mid s \in T_\gamma \}| = 1$.
  Then the following holds:
\begin{align*}
  T \vDash \theta^{=}(\vv{\alpha};\vv{\beta};\gamma) &\Leftrightarrow \forall (s,s') \in T_\gamma \times T_{\neg \gamma} : s(\vv{\alpha}) = s'(\vv{\beta})\\
  T \vDash \theta^{\neq}(\vv{\alpha};\vv{\beta};\gamma) &\Leftrightarrow \forall (s,s') \in T_\gamma \times T_{\neg\gamma} : s(\vv\alpha)\neq s'(\vv\beta)
\end{align*}
\end{lem}
\begin{proof}
  As $\theta^=$ is straightforward, let us consider $\theta^{\neq}$.

  For "$\Rightarrow$", by the formula, $T$ can be divided into $Y_1\cup \cdots \cup Y_n$ such that $Y_i \cap T_{\gamma} \neq \emptyset$ and additionally $Y_i$ satisfies the respective Boolean disjunction.
  Now let $s \in T_\gamma$ and $s' \in T_{\neg \gamma}$.
  For some $i \geq 1$, $s' \in Y_i$.
  Furthermore, there is $l \in \{\top,\bot\}$ such that $Y_i \vDash \E \gamma \land (\gamma \land (\alpha_i \equi l))\lor (\neg \gamma \land (\beta_i \nequi l))$.
  As $Y_i \vDash \E \gamma$, some $s^\star \in Y_i \cap T_\gamma$ exists, and we conclude $s(\alpha_i) = s^\star(\alpha_i) \neq s'(\beta_i)$.

  For "$\Leftarrow$", we divide $T$ into teams $Y_1 \cup \cdots \cup Y_n$ as follows.
  For every $i \in \{1,\ldots,n\}$, choose $l \in \{\top,\bot\}$ such that
  \[
  Y_i \dfn \{ s \in T_{\neg \gamma} \mid s(\beta_i) \neq s'(\alpha_i), s' \in T_\gamma \}\text{.}
  \]
  $Y_i$ is well-defined as $s'(\alpha_i)$ is constant for all $s' \in T_\gamma$.
  This is a split of $T$, as otherwise some $s \in T_{\neg \gamma}$ is left over with $s(\beta_i) = s'(\alpha_i)$ for all $i \in \{1,\ldots,n\}$, contradicting the assumption.
  Clearly $Y_i \vDash \E \gamma$, as $T_{\gamma}\neq \emptyset$ and $T_{ \gamma} \subseteq Y_i$.
  It remains to check that setting $l \dfn \top$ if $a_i = 1$ (resp.\ $l \dfn \bot$ if $a_i = 0$) renders $(\gamma \land (\alpha_i \equi l))\lor (\neg \gamma \land (\beta_i \nequi l))$ true in $Y_i$.
\end{proof}

With $\theta^=$ and $\theta^{\neq}$ we can now define the remaining atoms.
To define the condition $|\{ s(\vv\alpha) \mid s \in T_\gamma\}| = 1$ in a formula, we use $\gamma \hook \mathbf{1}_{\alpha}$, where $\mathbf{1}_{\alpha} \dfn \sneg \bot \land \bigwedge_{i=1}^n\dep{\alpha_i}$.
Let us call an assignment $s$ in $T_\gamma$ that is unique up to $\alpha$ a \emph{pivot}.

\subsection*{Exclusion atom.}

With the exclusion atom, we exemplify how the formula $\theta^=$ can be used.
A team $T$ violates the exclusion atom $\vv\alpha\mid\vv\beta$ if either
some assignment $s$ satisfies $\vv{\alpha} \equi \vv{\beta}$, or otherwise if $s(\vv{\alpha}) = s'(\vv{\beta})$ for distinct $s,s'$.
Assuming we are only in the latter case, however, $s$ and $s'$ must disagree on some $\alpha_i$, say, $s \vDash \alpha_i$ and $s' \vDash \neg \alpha_i$, since otherwise we again have $s'(\vv{\alpha}) = s(\vv\alpha) = s'(\vv\beta)$.
Taking now $\gamma\dfn \alpha_i$, we can with $\lor$ split off everything from $T_\gamma$ except the pivot $s$, retain the team $\{ s \} \cup T_{\neg \gamma}$, and search for $s'$ in $T_{\neg \gamma}$ with $\theta^=$.

We apply these ideas in the following formula which expresses $\sneg(\vv{\alpha} \mid \vv{\beta})$ and has length $\calO(n|\vv\alpha\vv\beta|)$.
\begin{align*}
  \varphi(\vv{\alpha};\vv{\beta}) \dfn \E(\vv{\alpha} \equi \vv{\beta}) \bor \bigbor_{\mathclap{\substack{i=1\\\gamma\in \{\alpha_i,\neg \alpha_i\}}}}^n  \Big( \top \lor ((\E\neg \gamma) \land (\gamma \hook \mathbf{1}_{\vv\alpha}) \land \theta^=(\vv{\alpha};\vv{\beta};\gamma)) \Big)
\end{align*}
\begin{lem}\label{lem:exc}
 $\sneg \vv{\alpha} \mid \vv{\beta} \equiv \varphi(\vv{\alpha};\vv{\beta})$.
\end{lem}
\begin{proof}
Suppose $T \nvDash \vv{\alpha} \mid \vv{\beta}$, so there are $s,s' \in T$ such that $s(\vv{\alpha}) = s'(\vv{\beta})$.
 First, if $s(\vv{\alpha}) = s'(\vv{\alpha})$, then $T \vDash \E( \vv{\alpha} \equi \vv{\beta})$ and we are done.
 Otherwise, $s$ and $s'$ disagree on some $\gamma \in \{\alpha_i,\neg \alpha_i\mid 1\leq i\leq n\}$ such that $s(\gamma) = 1$ and $s'(\gamma) = 0$.
 Then the split $(T \setminus \{s,s'\},\{s,s'\})$ satisfies the Boolean disjunct with index $\gamma$, as clearly $\{s,s'\}$ satisfies  $\E \neg \gamma$, $\gamma \hook \mathbf{1}_{\vv\alpha}$, and $\theta^=(\vv{\alpha};\vv{\beta};\gamma)$.
 For the other direction, assume that $T \vDash \varphi(\vv{\alpha};\vv{\beta})$.
 Then either $T \vDash \E (\vv{\alpha} \equi \vv{\beta})$ and we are done, or there exist $\gamma$ and some split $(S,U)$ of $T$ such that $U \vDash (\E \neg \gamma) \land (\gamma \hook \mathbf{1}_{\vv\alpha}) \land   \theta^=(\vv{\alpha};\vv{\beta};\gamma)$.
This implies $U_{\gamma},U_{\neg \gamma} \neq \emptyset$, so $s_1(\vv{\alpha}) = s_2(\vv{\beta})$ for some
$(s_1,s_2) \in U_{\gamma} \times U_{\neg \gamma}$.
\end{proof}

\subsection*{Inclusion atom.}

A team $T$ falsifies the inclusion atom $\vv\alpha \subseteq \vv\beta$ if there exists $s^\star \in T$ such that $s^\star(\vv{\alpha}) \neq s(\vv{\beta})$ for all $s \in T$.
In particular, some $s^\star \in T$ must exist such that $s^\star(\alpha_i) \neq s^\star(\beta_i)$ for some $i$.
Similar as for the exclusion atom, it suffices to compare $s^\star(\vv\alpha)$ and $s(\vv\beta)$ only for assignments $s$ such that $s(\beta_i) \neq s^\star(\beta_i)$, as $s(\beta_i) = s^\star(\beta_i)$ already ensures $s^\star(\vv\alpha) \neq s(\vv\beta)$.
Hence $s^\star$ is a pivot for $\gamma \dfn \beta_i$, and it suffices to compare pairs from $\{s^\star\} \times T_{\neg \gamma}$ with $\theta^{\neq}$.

The following formula expresses the negation of the inclusion atom $\vv{\alpha}\subseteq \vv{\beta}$ and has length $\calO(n|\vv\alpha\vv\beta|)$.
\begin{align*}
   \varphi(\vv{\alpha};\vv{\beta})  \dfn \bigbor_{\mathclap{\substack{i=1\\\gamma\in\{\beta_i,\neg \beta_i\}}}}^n \Big(\gamma \lor \Big(\big(\gamma \hook  ( (\alpha_i \nequi \beta_i) \land \mathbf{1}_{\vv\alpha})\big) \land \theta^{\neq}(\vv{\alpha};\vv{\beta};\gamma)\Big)\Big)
\end{align*}
\begin{lem}\label{lem:inc}
 $\sneg\vv{\alpha} \subseteq \vv{\beta} \equiv \varphi(\vv{\alpha};\vv{\beta})$.
\end{lem}
\begin{proof}
 Let $T \nvDash \vv{\alpha} \subseteq \vv{\beta}$.
We show that $T \vDash \varphi(\vv{\alpha};\vv{\beta})$.
 By definition, there is $s^\star \in T$ such that $s^\star(\vv{\alpha}) \neq s(\vv{\beta})$ for all $s \in T$.
 In particular, $s^\star(\alpha_i) \neq s^\star(\beta_i)$ for some $i \in \{1,\ldots,n\}$.
 Let $\gamma \in \{\beta_i,\neg \beta_i\}$ such that $s^\star(\gamma) = 1$, and consider the subteam $S \dfn \{ s^\star \} \cup T_{\neg \gamma}$ of $T$.
 We show that the Boolean disjunct with index $\gamma$ is satisfied by the split $(T\setminus S,S)$.
 Clearly, $T \setminus S \vDash \gamma$.
 Moreover, $S_{\gamma} = \{ s^\star\} \vDash (\alpha_i \nequi \beta_i) \land \mathbf{1}_{\vv\alpha}$.
Finally, $S\vDash \theta^{\neq}(\vv{\alpha};\vv{\beta};\gamma)$ holds
since $s^\star(\vv{\alpha}) \neq s(\vv{\beta})$ for all $s \in T_{\neg \gamma}$ by assumption.

\smallskip

Conversely, assume $T \vDash \varphi(\vv{\alpha};\vv{\beta})$ with $1 \leq i \leq n$ and $\gamma \in \{ \beta_i,\neg \beta_i \}$ chosen according to a satisfying disjunct of $\varphi(\vv{\alpha};\vv{\beta})$.
By the formula, $T$ can be divided into $X \cup S$ with $X \vDash \gamma$, $S_{\gamma} \vDash (\alpha_i \neq \beta_i) \land \mathbf{1}_{\vv\alpha}$, and $S \vDash \theta^{\neq}(\vv{\alpha};\vv{\beta};\gamma)$.
In particular $S_{\gamma} = \{s^\star\}$ for some $s^\star$.
We show that for all $s \in T$ we have $s^\star(\vv{\alpha}) \neq s(\vv{\beta})$, so $T \nvDash \vv{\alpha} \subseteq \vv{\beta}$.
For all $s \in T_{\neg \gamma}$, this follows since $T_{\neg \gamma} = S_{\neg\gamma}$ and $S \vDash \theta^{\neq}(\vv{\alpha};\vv{\beta};\gamma)$.
For all $s \in T_\gamma$, this follows since $s(\beta_i) = s^\star(\beta_i)$ (recall that $\gamma \in \{\beta_i,\neg\beta_i\}$) and $s^\star(\beta_i) \neq s^\star(\alpha_i)$.
\end{proof}

\subsection*{Independence atom.}

The independence atom $\vv\alpha\perp_{\vv\gamma}\vv\beta$ is a bit more complicated:
It is false if there are $s^\star,s^\circ \in T$ that agree on $\vv\gamma$, and any $s \in T$ disagrees either with $s^\star$ on $\vv\alpha\vv\gamma$ or with $s^\circ$ on $\vv\beta$.
We separate the team along two "axes", with $\delta$ and $\epsilon$, have one pivot ($s^\star$ or $s^\circ$) for each, and two occurrences of $\theta^{\neq}$.

The following formula expresses the negation of the conditional independence atom $\vv{\alpha} \perp_{\vv{\gamma}} \vv{\beta}$ and has length $\calO(n(n+m+k)|\vv\alpha\vv\beta\vv\gamma|)$, where $\vv{\alpha} = (\alpha_1,\ldots,\alpha_n)$, $\vv{\beta} = (\beta_1,\ldots,\beta_m)$, and $\vv{\gamma} = (\gamma_1,\ldots,\gamma_k)$.
\begin{align*}
  \varphi(\vv{\alpha};\vv{\beta};\vv{\gamma}) \dfn \bigbor_{\mathclap{\substack{\delta \in \{\alpha_i,\neg \alpha_i \mid 1 \leq i \leq n\}\\
  \epsilon \in \{\beta_j,\neg \beta_j\mid 1 \leq j \leq m\}}}} (\delta\lor \epsilon) \lor \Big( \big[ (\neg \delta \land \theta^{\neq}(\vv{\alpha}\vv{\gamma};\vv{\alpha}\vv{\gamma};\epsilon)) \lor (\neg \epsilon \land \theta^{\neq}(\vv{\beta};\vv{\beta};\delta))&\big]\\[-6mm]
  \; \land\, \E \delta \land \E \epsilon \land ((\delta \lor \epsilon) \hook & \;\mathbf{1}_{\vv\gamma}) \Big)
\end{align*}
\begin{lem}\label{lem:ind}
 $\sneg\vv{\alpha} \perp_{\vv{\gamma}} \vv{\beta} \equiv \varphi(\vv{\alpha};\vv{\beta};\vv{\gamma})$.
\end{lem}
\begin{proof}
For the direction from left to right, assume $T \nvDash \vv{\alpha} \perp_{\vv{\gamma}} \vv{\beta}$.
Then there are $s^\star,s^\circ \in T$ such that $s^\star(\vv{\gamma}) = s^\circ(\vv{\gamma})$, but for all $s \in T$ it holds either $s(\vv{\alpha}\vv{\gamma}) \neq s^\star(\vv{\alpha}\vv{\gamma})$ or $s(\vv{\beta}) \neq s^\circ(\vv{\beta})$.
In particular, there must be $i,j$ such that $s^\star(\alpha_i) \neq s^\circ(\alpha_i)$ and $s^\star(\beta_j) \neq s^\circ(\beta_j)$.
Let $\delta \in \{\alpha_i,\neg \alpha_i\}$ and $\epsilon\in \{\beta_j,\neg \beta_j\}$ such that $s^\star(\epsilon) = s^\circ(\delta) = 1$ and $s^\star(\delta) = s^\circ(\epsilon) = 0$.
In order to now satisfy the Boolean disjunct with index $\delta,\epsilon$, we define subteams
\begin{align*}
S &\dfn \{ s^\star \} \cup \{ s \in T \mid s(\delta) = s(\epsilon) = 0, s(\vv{\alpha}\vv{\gamma}) \neq s^\star(\vv{\alpha}\vv{\gamma}) \}\\
U &\dfn \{ s^\circ \} \cup \{ s \in T \mid s(\delta) = s(\epsilon) = 0, s(\vv{\beta}) \neq s^\circ(\vv{\beta}) \}
\end{align*}
of $T$.
We show that the (in fact strict) split $(T \setminus (S\cup U), S \cup U)$ satisfies the disjunction.
First, $T\setminus (S \cup U) \vDash \delta \lor \epsilon$ due to the fact that $T \setminus (S \cup U) \subseteq T_{\delta} \cup T_{\epsilon}$.
Furthermore, $S \cup U \vDash \E \delta \land \E \epsilon \land (\delta \lor \epsilon) \hook \mathbf{1}_{\vv\gamma}$, since ${(S\cup U)}_{\delta \lor \epsilon} = \{s^\star,s^\circ\}$.
For the part in brackets, consider the (again strict) split $(S,U \setminus S)$ of $S \cup U$.
Again, clearly $S \vDash \neg \delta$ and $U \vDash \neg \epsilon$.
Finally, both $S \vDash \theta^{\neq}(\vv{\alpha}\vv{\gamma};\vv{\alpha}\vv{\gamma};\epsilon)$ and $U \vDash \theta^{\neq}(\vv{\beta};\vv{\beta};\delta)$ hold.

For the other direction, assume $T \vDash \varphi(\vv{\alpha};\vv{\beta};\vv{\gamma})$ with the Boolean disjunction satisfied with indices $\delta \in \{\alpha_i,\neg \alpha_i\mid 1 \leq i \leq n\}$ and
$\epsilon \in \{\beta_j,\neg \beta_j\mid 1 \leq j \leq m\}$.
Then $T$ can be divided into $X \cup S \cup U$ where
  \begin{itemize}
    \item $X \vDash \delta \lor \epsilon$,
    \item $S \vDash \neg \delta \land \theta^{\neq}(\vv{\alpha}\vv{\gamma};\vv{\alpha}\vv{\gamma};\epsilon)$,
    \item $U \vDash \neg \epsilon \land \theta^{\neq}(\vv{\beta};\vv{\beta};\delta)$
    and
    \item $S \cup U \vDash \E\delta \land \E\epsilon \land ((\delta \lor \epsilon) \hook \mathbf{1}_{\vv\gamma})$.
  \end{itemize}
  By the final line, assignments $s^\star \in S_{\epsilon}$ and $s^\circ \in U_{\delta}$ exist.
  Now, for the sake of contradiction, suppose that $T \vDash \vv{\alpha}\perp_{\vv{\gamma}}\vv{\beta}$.
  As $S_\epsilon \cup U_\delta \vDash \mathbf{1}_{\vv\gamma}$ and hence $s^\star(\vv{\gamma}) = s^\circ(\vv{\gamma})$, due to independence, another assignment $s \in T$ must exist such that $s(\vv{\alpha}\vv{\gamma}) = s^\star(\vv{\alpha}\vv{\gamma})$ and $s(\vv{\beta}) = s^\circ(\vv{\beta})$.

However, $s \notin X$, since $s(\vv{\alpha}) = s^\star(\vv{\alpha})$ implies $s \nvDash \delta$ and $s(\vv{\beta}) = s^\circ(\vv{\beta})$ implies $s \nvDash \epsilon$.
Consequently, $s \in S \cup U$.
For this reason, either $s(\vv{\alpha}\vv{\gamma}) \neq s^\star(\vv{\alpha}\vv{\gamma})$, or $s(\vv{\beta}) \neq s^\circ(\vv{\beta})$, contradiction to $s(\vv{\alpha}\vv{\gamma}) = s^\star(\vv{\alpha}\vv{\gamma})$ and $s(\vv{\beta}) = s^\circ(\vv{\beta})$.
\end{proof}

\subsection*{Anonymity atom.}

Finally, the following formula expresses the negation of the unary anonymity atom $\vv{\alpha}\anon \beta$ and has length $\calO(n|\beta|+|\vv\alpha|)$.

Roughly speaking, the anonymity atom $\vv{\alpha}\anon \beta$ is false if there is $s^\star \in T$ such that no $s \in T$ with identical $\vv\alpha$ but different $\beta$ exists, or in other words, all $s \in T$ with different $\beta$ are also different in $\vv\alpha$.
So we can directly let $\gamma \dfn \beta$ or $\gamma \dfn \neg \beta$, pick $s$ as pivot, and apply $\theta^{\neq}$ to $\alpha$:
\begin{align*}
   \varphi(\vv{\alpha};\beta)  \dfn \bigbor_{\mathclap{\gamma \in \{\beta,\neg \beta\}}} \Big( \gamma \lor \big( (\gamma \hook \mathbf{1}_{\vv\alpha}) \land \theta^{\neq}(\vv{\alpha};\vv{\alpha};\gamma) \big) \Big)
\end{align*}
\begin{lem}\label{lem:anon-unary}
  $\sneg\vv{\alpha} \anon \beta \equiv \varphi(\vv{\alpha};\beta)$.
\end{lem}
\begin{proof}
  Suppose $T \nvDash \vv{\alpha} \anon \beta$.
  Then there is $s^\star \in T$ such that $s(\vv{\alpha}) = s^\star(\vv{\alpha})$ implies $s(\beta) = s^\star(\beta)$ for all $s \in T$.
  Let $\gamma \in \{\beta,\neg \beta\}$ such that $s^\star \vDash \gamma$, and consider the split $(T \setminus S,S)$ of $T$ defined by $S \dfn \{ s^\star \} \cup T_{\neg\gamma}$.
  Then $T \setminus S \vDash \gamma$.
  Moreover, $S \vDash \gamma \hook \mathbf{1}_{\vv\alpha}$ and $S \vDash \theta^{\neq}(\vv{\alpha};\vv{\alpha};\gamma)$, since $S_{\gamma} = \{ s^\star\}$ and $s(\vv{\alpha}) \neq s^\star(\vv{\alpha})$ for all $s \in S_{\neg \gamma}$.

  For the other direction, suppose there is $\gamma \in \{\beta,\neg \beta\}$ such that $S \vDash \gamma$ and $U \vDash (\gamma \hook \mathbf{1}_{\vv\alpha}) \land \theta^{\neq}(\vv{\alpha};\vv{\alpha};\gamma)$ for some split $(S,U)$ of $T$.
  Then there exists $s^\star \in U_\gamma$ such that $s^\star(\vv{\alpha}) \neq s(\vv{\alpha})$ for all $s \in U_{\neg \gamma}$.
  Clearly, now $s^\star(\beta) = s(\beta)$ for all $s \in S \cup U_{ \gamma}$, so ultimately $s^\star(\beta) = s(\beta)$ or $s^\star(\vv{\alpha}) \neq s(\vv{\alpha})$ for all $s \in T$, hence $T \nvDash \vv{\alpha} \anon \beta$.
\end{proof}

In the first-order setting, Rönnholm~\cite[Remark 2.31]{RRthesis} demonstrated that the general anonymity atom can be expressed via the unary anonymity atom and the splitting disjunction.
In the lemma below, we show that this can also be done via strict splitting.
This yields a formula expressing $\vv{\alpha} \anon \vv{\beta}$ of length $\calO(n|\vv\beta|+m|\vv\alpha|)$.

\begin{lem}\label{lem:anon}
  The following formulas are equivalent:
  \begin{enumerate}
    \item[(1)] $\vv{\alpha} \anon \vv{\beta}$,
    \item[(2)] $\bigvee_{i=1}^m \vv{\alpha}\anon\beta_i$,
    \item[(3)] $\dot\bigvee_{i=1}^m \,\vv{\alpha}\anon\beta_i$.
  \end{enumerate}
\end{lem}
\begin{proof}
For (2) $\Rightarrow$ (1), we follow Rönnholm~\cite{RRthesis}.
Suppose $T \vDash \bigvee_{i=1}^m \vv{\alpha}\anon\beta_i$ via the split of $T$ into $Y_1 \cup \cdots \cup Y_m$, where $Y_i \vDash \vv{\alpha}\anon\beta_i$.
To see that $T \vDash \vv{\alpha} \anon \vv{\beta}$, let $s \in T$ be arbitrary.
For some $i$, now $s \in Y_i$.
Consequently, there is $s' \in Y_i$ such that $s(\vv{\alpha}) = s'(\vv{\alpha})$ but $s(\beta_i) \neq s'(\beta_i)$.
But as $Y_i \subseteq T$ and $s$ was arbitrary, (1) follows.

The step (3) $\Rightarrow$ (2) is clear, since every strict split of a team is a split.

It remains to show (1) $\Rightarrow$ (3).
Here, we adapt the proof of Rönnholm~\cite{RRthesis} for $\sor$.
Suppose that $T \vDash \vv{\alpha} \anon \vv{\beta}$ holds.
Define subteams $Y_i$ of $T$ by
\[
Y_i \dfn \Set{ s \in T \mid \exists s' \in T : s'(\vv{\alpha}) = s(\vv{\alpha}) \text{ but }s(\beta_i) \neq s'(\beta_i) }\text{,}
\]
as in the proof of Rönnholm~\cite{RRthesis}, but additionally define teams $Z_i \dfn Y_i \setminus \bigcup_{j < i} Y_{j}$ for $1 \leq i \leq m$, where $Y_0 \dfn \emptyset$.
We show that $Z_1 \cup \cdots \cup Z_m$ forms a strict split of $T$.
The sets $Z_1,\ldots,Z_m$ are pairwise disjoint, as $Z_i \subseteq Y_i$ but $Z_{j} \cap Y_i = \emptyset$ when $i < j$.
Next, let $s \in T$ be arbitrary.
Define
\[
I \dfn \{ i \in \{1,\ldots,m\} \mid \exists s' \in T : s(\vv\alpha) = s'(\vv\alpha) \text{ but } s(\beta_i) \neq s'(\beta_i)\}\text{.}
\]
By assumption (1), $I$ is non-empty and hence contains a minimal element $i$.
But then $s \in Y_i \setminus \bigcup_{j<i}Y_{j} = Z_i$.
Consequently, $T = \bigcup_{i=1}^m Z_i$.

Finally, we need to show that $Z_i \vDash \vv\alpha \anon \beta_i$.
For this, let now $s \in Z_i$ be arbitrary.
By definition of $Z_i$, there exists $s' \in T$ with $s(\vv\alpha) = s'(\vv\alpha)$ and $s(\beta_i) \neq s'(\beta_i)$.
It suffices to show that $s' \in Z_i = Y_i \setminus \bigcup_{j<i}Y_j$.
As $s' \in Y_i$ follows from the definition of $Y_i$, assume $s' \in Y_j$ for some $j < i$.
Then by symmetry also $s \in Y_j$, contradiction to $s \in Z_i$.
Hence $s' \notin Y_j$ for all $j < i$, so $s' \in Z_i$.
\end{proof}

With the negations of dependency atoms definable in $\PL(\{ \bor, \land, \lor \})$, it is an easy corollary that the atoms themselves are definable when additionally the strong negation $\sneg$ is available.
In the next theorem, we prove this, generalize the part on the anonymity atom $\anon$, and furthermore expand the results to also work with $\sor$, which we previously considered only for the downward closed atoms $\dep{\cdot,\cdot}$ and $\mid$ in Theorem~\ref{thm:upper-bounds}.

\begin{thm}\label{thm:upper-bounds-with-negation}
  Let $\Sigma = \{ \sneg, \land, \lor \}$ or $\Sigma = \{ \sneg, \land, \sor \}$.
  Let $\vv\alpha = (\alpha_1,\ldots,\alpha_n)$, $\vv\beta = (\beta_1,\ldots,\beta_m)$, and $\vv\gamma = (\gamma_1,\ldots,\gamma_k)$ be tuples of purely propositional formulas.
  \begin{itemize}
    \item The dependence atom $\dep{\vv\alpha;\vv\beta}$ is equivalent to a $\PL(\Sigma)$-formula of length $\calO(|\vv\alpha\vv\beta|)$.
    \item The exclusion atom $\vv{\alpha} \mid \vv{\beta}$ is equivalent to a $\PL(\Sigma)$-formula of length $\calO(n|\vv\alpha\vv\beta|)$.
    \item The inclusion atom $\vv\alpha\subseteq \vv\beta$ is equivalent to a $\PL(\Sigma)$-formula of length $\calO(n|\vv\alpha\vv\beta|)$.
    \item The conditional independence atom $\vv\alpha\perp_{\vv\gamma}\vv\beta$ is equivalent to a $\PL(\Sigma)$-formula of length $\calO(n(n+m+k)|\vv\alpha\vv\beta\vv\gamma|)$.
    \item The anonymity atom $\vv\alpha\anon \vv\beta$ is equivalent to a $\PL(\Sigma)$-formula of length $\calO(n|\vv\beta|+m|\vv\alpha|)$.
  \end{itemize}
  Furthermore, all these formulas are logspace-computable.
\end{thm}
\begin{proof}
We essentially take the formulas of Theorem~\ref{thm:upper-bounds} (and Lemma~\ref{lem:anon} for the anonymity atom) and add a Boolean negation in front of them.
For $\Sigma = \{ \sneg, \land, \lor\}$, the only remaining thing to do is to rewrite $\bor$ via $\land$ and $\sneg$.

For $\Sigma = \{ \sneg, \land, \sor\}$, we must also remove all occurrences of $\lor$ and use only $\sor$.
We see that this comes down to expressing the subformulas $\theta^=$ and $\theta^{\neq}$ in $\{\sneg, \land, \sor\}$.
In $\theta^=$, the lax splitting $\lor$ can equivalently be replaced by $\sor$ due to Proposition~\ref{prop:strict-lax}, as any occurrence of $\lor$ has at least one purely propositional argument.
The same does not hold for $\theta^{\neq}$, but it is easy to see that $\theta^{\neq}(\vv{\alpha};\vv{\beta};\gamma)$ can be replaced by
\begin{align*}
\gamma \ovee \sneg(\top \sor (\E \gamma \land \E\neg\gamma \land \theta^=(\vv{\alpha};\vv{\beta};\gamma)))\text{.}\tag*{\qedhere}
\end{align*}
\end{proof}

\subsection{Upper bounds for parity}

\smallskip

Next, we again consider the parity of the cardinality of teams, \ie, is there a formula that is true precisely on teams with even cardinality?
This differs from the other considered team properties in that both the property and its negation have exponential lower bounds in $\PL(\{\bor,\land,\lor,\sor\})$ (see Theorem~\ref{thm:lower-bounds}).
Nevertheless, we show that it is polynomially definable when linearly many negations are nested inside the formula, which was not necessary for the results of Theorem~\ref{thm:upper-bounds-with-negation}.

\begin{thm}
  Let $|\Phi| = n$.
  The class of $\Phi$-teams of odd cardinality is defined by a $\PL(\land,\sneg,\sor)$-formula of length $\calO(n^2)$.
\end{thm}

We write $\dep{X}$, for a finite set $X \subseteq \Phi$ of propositions, as abbreviation for $\bigwedge_{p \in X}\dep{p}$.
Based on this, the formula $\mathbf{1} \dfn \sneg\bot \land \dep{\Phi}$ defines singletons, that is, a $\Phi$-team $T$ satisfies $\mathbf{1}$ iff $|T| = 1$.
The formula expressing odd cardinality is now recursively defined as follows:
\begin{align*}
  \varphi(\+) \;&\dfn \mathbf{1}\\
  \varphi(p\vv{q}) \;& \dfn \mathbf{1} \sor \sneg\Big(\big[\mathbf{1} \bor  \big(\sneg\dep{p}\land\big(\mathbf{1} \sor \dep{p} \big)\big) \big]  \sor [\dep{p}\land \sneg\varphi(\vv{q})] \Big)
\end{align*}
We prove its correctness in the lemma below.
The rough idea is that a team is even precisely if $T_p$ and $T_{\neg p}$ are either both even or both odd, regardless of which proposition $p$ is.

\begin{lem}
  Let $T \in \teams(\Phi)$ and let $\vv{q}$ list all propositions in $\Phi$.
  Then $T \vDash \varphi(\vv{q})$ if and only if $|T|$ is odd.
\end{lem}
\begin{proof}
The proof is by induction on $|\vv{q}|$.
Since the domain of $T$ exceeds the arguments of the recursive subformulas $\phi$, we prove the following stronger statement.
Let $\vv{q} = (q_1,\ldots,q_m)$.
Then, for any $\Phi$-team $S$ satisfying $\dep{\Phi \setminus \{q_1,\ldots,q_m\}}$, it holds that that $S \vDash \varphi(\vv{q})$ if and only if $|S|$ is odd.
The base case is clear as the only $\emptyset$-teams are $\emptyset$ and $\{\emptyset\}$.

We proceed with the inductive step, and first provide some intuition.
The crucial subformula is
  \[
  \psi \dfn \big[\mathbf{1} \bor  \big(\sneg\dep{p}\land\big(\mathbf{1} \sor \dep{p} \big)\big) \big]  \sor [\dep{p}\land \sneg\varphi(\vv{q})]\text{.}
  \]
  We will show below that it is true iff at least one of $|S_p|$ and $|S_{\neg p}|$ is odd.
  Then $\sneg \psi$ means that both $|S_p|$ and $|S_{\neg p}|$ are even.
  This is sufficient for $|S|$ to be even but of course not necessary.
  However, the following holds:
  $|S|$ is odd precisely when we can remove one assignment $S$ such that afterwards both $|S_p|$ and $|S_{\neg p}|$ are even.
  Hence, oddness is defined by $\mathbf{1} \sor \sneg \psi$.

  Intuitively, $\psi$ allows to split off an even subteam of either $S_p$ or $S_{\neg p}$ by \ldots$\+\sor \, (\dep{p} \land \sneg\phi(\vv{q}))$, reducing either $S_p$ or $S_{\neg p}$, depending on which is odd, to a singleton.
  Afterwards the team then satisfies $\mathbf{1} \bor  \sneg\dep{p}\land\big(\mathbf{1} \sor \dep{p} \big)$.
  We prove this formally, \ie, that $S \vDash \psi$ iff $|S_p|$ or $|S_{\neg p}|$ is odd.
\begin{itemize}[leftmargin=1cm]
  \item["$\Rightarrow$"]
  Suppose $S \vDash \psi$ via the strict split $(U,V)$ such that $V \vDash \dep{p} \land \sneg\varphi(\vv{q})$, and either $U \vDash \mathbf{1}$ or $U \vDash \sneg\dep{p} \land (\mathbf{1}\sor \dep{p})$.
  Note that $|V|$ is even by induction hypothesis.
  We distinguish the two possible cases for $U$.
  \begin{itemize}
    \item $U \vDash \mathbf{1}$:
    Then $U,V \vDash \dep{p}$.
    Additionally,
    Both $U$ and $U \cup V$ have odd size, and one of them equals $S_p$ or $S_{\neg p}$, depending on whether $U$ and $V$ agree on $p$ or not.

    \item $U \vDash \sneg\dep{p}\land (\mathbf{1}\sor \dep{p})$:
   Due to symmetry, we can assume $V \subsetneq S_p$ and $S_{\neg p} \subsetneq U$.
  By the formula, $U$ has a strict split $(X,Y)$ such that $|X| = 1$ and $Y \vDash \dep{p}$.
  Let $Z = S_p \setminus V$.
  Either $Z \subseteq X$ or $Z \subseteq Y$, as $X$ and $Y$ do not agree on $p$, but each is constant in $p$.
  If $Z \subseteq X$, then $Z = X$ and $|V \cup X| = |S_p|$ is odd and we are done.
  If $Z \subseteq Y$, then $S_{\neg p} \subseteq X$, hence $S_{\neg p} = X$ and $|S_{\neg p}|$ is odd.
  \end{itemize}
  \item["$\Leftarrow$"]
  \Wloss $|S_p|$ is odd.
  Pick $s \in S_p$ arbitrarily and consider the split $(S_{\neg p}\cup \{s\}, S_{p} \setminus \{s\})$ of $S$.
  For the second component, $S_p \setminus \{s\} \vDash \dep{p} \land \sneg\varphi(\vv{q})$ by induction hypothesis.
  For the first component, either $S_{\neg p}$ is empty and $S_{\neg p} \cup \{s\} \vDash \mathbf{1}$, or $S_{\neg p}$ is non-empty and $S_{\neg p} \cup \{s\} \vDash \sneg\dep{p} \land (\mathbf{1} \sor \dep{p})$.
  In both cases, $S \vDash \psi$.\qedhere
\end{itemize}
\end{proof}

We have shown an exponential lower bound for parity in the existential fragment.
For the matching upper bound, the following formulas define parity by mutual recursion:
\begin{align*}
  \phi^\mathrm{even}(\+) & \dfn \bot\\
  \phi^\mathrm{odd}(\+) & \dfn \nempty\\
  \phi^\mathrm{even}(p\vv{q}) &\dfn \big((p \land \phi^\mathrm{odd}(\vv{q})) \lor (\neg p \land \phi^{\mathrm{odd}}(\vv{q}))\big) \bor \big( (p \land \phi^\mathrm{even}(\vv{q})) \lor (\neg p \land \phi^{\mathrm{even}}(\vv{q})) \big)\\
  \phi^\mathrm{odd}(p\vv{q}) &\dfn \big((p \land \phi^\mathrm{odd}(\vv{q})) \lor (\neg p \land \phi^{\mathrm{even}}(\vv{q})) \big) \bor \big( (p \land \phi^\mathrm{even}(\vv{q})) \lor (\neg p \land \phi^{\mathrm{odd}}(\vv{q}))\big)
\end{align*}

\smallskip

\begin{thm}\label{thm:parity-exp}
  Let $|\Phi| = n$.
  If $\Sigma = \{\land,\bor,\lor\}$ or $\Sigma = \{\land,\bor,\sor\}$, then the class of $\Phi$-teams of odd resp.\ even cardinality is definable by a $\PL(\Sigma)$-formula of length $2^{\calO(n)}$.
\end{thm}
\begin{proof}
  First of all, observe that the formula $(p \land \phi) \lor (\neg p \land \phi')$ is equivalent to $(p \land \phi) \sor (\neg p \land \phi')$ for all $\phi,\phi'$ and propositions $p$, since any split satisfying the former formula is necessarily strict.
  As a consequence, it suffices to consider $\Sigma = \{\land,\bor,\lor\}$.

  Let $\Phi = \{ p_1, \ldots, p_n\}$, and $T$ a $\Phi$-team.
  Let $\vv{p} = (p_1,\ldots,p_n)$ list all variables in $\Phi$.
  We prove by induction on $n$ that $T \vDash \phi^\mathrm{even}(\vv{p})$ iff $|T|$ is even, and $T \vDash \phi^\mathrm{odd}(\vv{p})$ iff $|T|$ is odd.

  First, if $\Phi = \emptyset$, then either $T = \emptyset$ and $T \vDash \bot = \phi^\mathrm{even}(\+)$, or $T = \{ \emptyset \}$ and $T \vDash \nempty = \phi^\mathrm{odd}(\+)$.
  For the inductive step, observe that $|T|$ is even iff $|T_p|$ and $|T_{\neg p}|$ have equal parity, and is odd iff they have different parity, where $p \in \Phi$ is an arbitrary proposition.
  Furthermore, $T_p$ and $\restr{T_p}{(\Phi\setminus \{p\})}$ have the same cardinality (the same goes for $T_\neg p$).
  Additionally, $T_p$ and $\restr{T_p}{(\Phi\setminus \{p\})}$ satisfy the same $\PL(\Phi\setminus \{p\},\Sigma)$-formulas by Proposition~\ref{prop:locality}.
  Hence the equivalence immediately follows by induction hypothesis.
\end{proof}

 \subsection{Modal team logic}

In this final section, we consider \emph{modal team logic} $\MTL$, introduced by Müller~\cite{mueller14}, which extends both classical modal logic $\ML$ and propositional team logic $\PL(\{\land,\sneg,\lor\})$.
Beginning with modal dependence logic by Väänänen~\cite{vaa08}, several atoms of dependency have been transferred from the first-order setting also to the modal setting~(cf.~\cite{EbbingHMMVV13,KontinenMSV17,HellaS15}).
Using the results of this paper, we show that the computational complexity of $\MTL$ does not change if it is augmented with any of the dependency atoms we considered before.

For each $k \geq 0$, we define the function $\exp_k$ as $\exp_0(n) \dfn n$ and $\exp_{k+1}(n) \dfn 2^{\exp_k(n)}$.

For $k \geq 0$, $\ATIMEALT{\exp_k}$ is the class of problems decidable by an alternating Turing machine (see~\cite{alternation}) with at most $p(n)$ alternations and runtime at most $\exp_k(p(n))$, for a polynomial $p$.
Likewise, $\TOWERPOLY$ is the class of problems that are decidable by a deterministic Turing machine in time $\exp_{p(n)}(1)$ for some polynomial $p$.

The syntax of $\MTL$ is given by the following grammar, where $p$ is an atomic proposition:
\begin{align*}
 \phi & \ddfn \top \mid \bot \mid p \mid \neg p \mid \sneg \phi \mid \phi  \land \phi  \mid \phi  \lor \phi \mid \Box\phi \mid \Diamond\phi \text{,}
\end{align*}
Observe that classical modal logic $\ML$ is the $\sneg$-free fragment of $\MTL$.
Let $\md(\phi)$ denote the \emph{modal depth} of $\phi$, \ie, the nesting depth of $\Diamond$ and $\Box$ inside $\phi$.
A \emph{Kripke structure over $\Phi$}, where $\Phi$ is a set of propositions, is a tuple $K = (W,R,V)$ where $(W,R)$ is a directed graph and $V \colon \Phi \to 2^W$.
A \emph{team} in $K$ is a subset of $W$.
Let $RT \dfn \{ v \mid (w,v) \in R, w \in T \}$ and $R^{-1}T \dfn \{ w \mid (w,v) \in R, v \in T \}$.
The set $\Prop(\phi)$ is defined as for propositional logic.
$\MTL$-formulas $\phi$ are evaluated as follows on pairs $(K,T)$, where $K$ is a Kripke structure over some set $\Phi' \supseteq \Prop(\phi)$ of propositions and $T$ is a team in $K$:
\begin{alignat*}{3}
  & (K, T) \vDash p           &  & \Leftrightarrow\; T \subseteq V(p)\text{ for }p \in \Phi\text{,}                           \\
  & (K, T) \vDash \neg p      &  & \Leftrightarrow\; T \cap V(p) = \emptyset \text{ for }p \in \Phi\text{,} \text{,}          \\
  & (K,T)\vDash \Diamond \psi &  & \Leftrightarrow\; \exists T' \subseteq RT : T \subseteq R^{-1}T' \text{ and }(K, T')\vDash \psi \\
  & (K,T)\vDash \Box\psi      &  & \Leftrightarrow\;(K,RT) \vDash \psi\text{,}
\end{alignat*}
with $\land,\sneg,\top$ and $\bot$ analogously to propositional logic.
An $\MTL$-formula $\phi$ is \emph{satisfiable} (\emph{valid}) if $(K,T)\vDash \phi$ for some (every) Kripke structure $K$ over $\Prop(\phi)$ and team $T$ in $K$.
The \emph{model checking problem} is, given $\phi \in \MTL$ and a Kripke structure with team $(K,T)$, to decide whether $(K,T) \vDash \phi$.

The modal atoms of dependence $\dep{\vv\alpha;\vv\beta}$, independence $\vv\alpha\perp_{\vv\beta}\vv\gamma$, inclusion $\vv\alpha \subseteq \vv\beta$,
exclusion $\vv\alpha\mid \vv\beta$, and anonymity $\vv\alpha \anon \vv\beta$, are defined completely analogous as the propositional variants (cf.\ p.~\pageref{p:atoms}), but with $\vv\alpha,\vv\beta,\vv\gamma$
being tuples of $\ML$-formulas instead of $\PL$-formulas.

\begin{thm}
  For $\MTL$ extended by the atoms $\dep{\cdot,\cdot}$, $\perp_c$, $\subseteq$, $\mid$, and $\anon$,
 \begin{itemize}
  \item satisfiability and validity is $\TOWERPOLY$-complete,
  \item satisfiability and validity for modal depth at most $k$ is $\ATIMEALT{\exp_k}$-complete,
  \item model checking is $\PSPACE$-complete,
 \end{itemize}
 with respect to logspace-reductions.
\end{thm}
\begin{proof}
  For the logic without any atoms, the complexity was shown by Müller~\cite{mueller14} and Lück~\cite{Luck18}.
  The upper bounds of Theorem~\ref{thm:upper-bounds} immediately carry over to $\MTL$, so we can substitute every such atom by a polynomially long equivalent $\MTL$-formula.
\end{proof}

\section{Conclusion}

In this paper, we classified common atoms of dependency with respect to their succinctness in various fragments of propositional team logic.
We showed that the negations of these atoms all can be polynomially expressed in the positive fragment of propositional team logic, while the atoms themselves can only be expressed in this fragment in formulas of exponential size.
This implies polynomial upper bounds for the atoms in full propositional team logic with unrestricted contradictory negation.
For the lower bounds, we adapted formula size games to the team semantics setting, and refined the approach with the notion of upper dimension.

In further research, comparing the atoms of dependency in terms of succinctness could be interesting.
For example, do the lower bounds for the inclusion atoms still hold if we consider the positive fragment together with dependence atoms?
Adding moves corresponding to atoms of dependency to the formula size game would enable looking into the relative succinctness.

\section*{Acknowledgment}
\noindent{}The authors want to thank Lauri Hella and Raine Rönnholm for pointing out references essential for this work, and Juha Kontinen for being the initiating force behind it.
We also thank the anonymous referees for their very detailed comments and suggestions on how to simplify several proofs.
This work was supported in part by a joint grant of the DAAD (57348395) and the Academy of Finland (308099).

\bibliographystyle{alpha}
\bibliography{PLsuccinct}

\end{document}